\documentclass[aip,jmp,superscriptaddress]{revtex4-1}
\usepackage{amsfonts}
\usepackage{graphicx}
\usepackage{amsmath}
\usepackage{amsthm}
\usepackage{bbm}
\usepackage{hyperref}
\newcommand{\oneover}[1]{\ensuremath{\frac{1}{#1}}}
\newcommand{\hil}{\ensuremath{\mathcal{H}}}
\newcommand{\ket}[1]{\ensuremath{|#1\rangle}}
\newcommand{\bra}[1]{\ensuremath{\langle #1|}}
\newcommand{\braket}[2]{\ensuremath{\langle #1 | #2 \rangle}}
\newcommand{\ave}[1]{\ensuremath{\left\langle#1\right\rangle}}

\newcommand{\tr}[1]{\ensuremath{\operatorname{Tr}\left[#1\right]}}
\newcommand{\trd}[2]{\ensuremath{D\left(#1,#2\right)}}
\newcommand{\norm}[1]{\ensuremath{\left\|#1\right\|}}

\newcommand{\misub}[2]{\ensuremath{N^{(#1)}_{#2}}}

\begin{document}

\title{Achieving the Holevo bound via a bisection decoding protocol}

\author{Matteo Rosati and Vittorio Giovannetti}
\affiliation{NEST, Scuola Normale Superiore and Istituto di Nanoscienze-CNR, Piazza dei Cavalieri 7, I-56126 Pisa, Italy}

 \begin{abstract}
  We present a new decoding protocol to realize transmission of classical information 
through a quantum channel at asymptotically maximum capacity, achieving the 
Holevo bound and thus the optimal communication rate. At variance with previous proposals, our scheme recovers the message bit by bit, making use of a series  ``yes-no'' measurements, organized in
bisection fashion, thus determining which codeword was sent in $\log_2N$ steps, $N$ 
being the number of codewords. 
\end{abstract}

\maketitle

\section{Introduction}
One of the main achievements in quantum information theory has been the 
development of a generalization of Shannon's theory for quantum communication\cite{WildeBOOK}. In particular, the 
Holevo bound \cite{HolevoBOOK,holevo1} sets a limit on the rate of reliable transmission of classical information 
through a quantum channel, which is also achievable in the asymptotic limit of infinitely long sequences 
\cite{holevo2,schumawest,holevo3,winter,oga,hauswoot,oganaga,hayanaga,hayashi,seq1,seq2,sen}. 
Consequently, via proper optimization and regularization \cite{hastings}, it provides the 
quantum analog of the Shannon classical capacity formula.\\
The original proof~\cite{holevo2,schumawest} was carried out by extending to the quantum regime the concept of 
typical subspaces used in Shannon communication theory \cite{classicalinfo,SchumaTyp}. A crucial point is the 
choice of a proper POVM which allows Bob to identify the right message with small 
error probability. The first explicit detection scheme used in this context is a one-step collective-measurement POVM known as 
\emph{Pretty Good Measurement} (PGM) \cite{schumawest,hauswoot}, highly effective theoretically but not 
easily realizable in practice. \\
Following the proof of Ogawa and Nagaoka \cite{oga,oganaga}, Hayashi and Nagaoka \cite{hayanaga}, which establishes 
a connection with the quantum-hypotesis-testing problem \cite{qhyptest}, the possibility of 
asymptotically achieving the bound through a series of ``yes-no'' projective measurements was investigated \cite{seq1,seq2,sen}.
This sequential protocol checks whether the received state resides in the typical 
subspace of a given codeword, for each codeword in the code, until it receives a 
positive answer or else declares failure. The ``yes-no'' question is asked, for 
each codeword, by applying the projector on its typical subspace and thus makes 
the decoding protocol more suited for practical implementations than the PGM. 
Indeed a design for an explicit and structured optical receiver was proposed \cite{wildeguha1,wildeguha2}, 
which used this protocol, with applications both to optical communication and quantum reading. In 
particular, for a lossy bosonic channel~\cite{HOLWER} (a model most commonly used to represent 
realistic fiber and free-space communication) it was shown that the sequential 
decoder can be built with gaussian displacement operators and vacuum-or-not 
measurements~\cite{sen,qhyptest,LOSSY}.
An alternative, near-explicit approach, for capacity-achieving classical-quantum communication 
was also recently  developed by Wilde and Guha\cite{polarWildeGuha},
adapting to the quantum scenario the classical polar coding introduced by Arikan\cite{Arikan}. In particular, 
making use  of optimal Helstrom
  measurements in the quantum-hypotesis-testing procedure and of Sen's non-commutative union bound\cite{sen},
 they proposed an encoding technique which realizes channel polarization and 
 consequently introduced a quantum successive cancellation decoder. 
  Later work modified such decoding strategy to obtain a partially non-collective measurement\cite{wildeHayden}
 and  extended polar coding to private and quantum communication through arbitrary qubit 
  channels\cite{privaQuant,wildeRen1,wildeRen2}. The relevance of this  approach is  associated with the fact that, at variance with other proposals~\cite{holevo2,schumawest,seq1,seq2,sen},
    it allows optimal decoding with a linear 
  (in the amount of bits) number of collective measurements.\\

In this paper we propose a bisection decoding scheme for classical communication through a quantum channel 
and show that it achieves the maximum capacity in the asymptotic limit of infinitely 
long codewords, providing yet an alternative proof of the attainability of the 
Holevo bound. While being inspired to the sequential decoding algorithm~\cite{seq1,seq2,sen}, 
analogously  to  Refs.~\onlinecite{wildeguha1,wildeguha2} our scheme exhibits an exponential advantage in the number of  measurements which have to be performed in order to recover
the message: specifically if the sequential method is built on $O(N)$ concatenated    ``yes-no'' detections, where $N$ is the number of codewords, 
the bisection method only requires $\log_2N$ of such ``{atomic}" steps, thus scaling linearly with the number of bits $n$ which one wishes to transmit. 
We stress however that, being our individual detections explicitly many-body operations, 
 at present we have no evidence in support of the fact that such advantage could be translated in a decoding scheme
which is efficient from the computational point of view, i.e. in terms of the number of quantum gates  one has to 
apply to the received string of quantum information carriers.A similar problem arises also in the case of polar codes (see e.g. 
Ref.~\onlinecite{wildeHayden}), and it is caused by the lack of an explicit implementation (or at least of an estimate of its complexity) of the ``{atomic}'' steps involved in the two protocols, i.e. the ``yes-no'' set  detections for the present method and the Helstrom measurement for polar coding. Still we believe that our method can be of some interest 
as it widens the class of known decoding strategies which are asymptotically optimal, increasing hence the chances of identifying at least one which 
is suitable to implementations. 
In this respect it is also worth noticing that the proposed scheme exhibits the nontrivial advantage of gaining a bit of information at each step of the procedure, a feature 
which may be extremely appealing when dealing with faulty decoders, as it  allows  partial identification of
    the transmitted message even in the presence of subsequent detection failures.

As in all the  previous works on the subject, in our derivation we heavily rely on the structure of typical projectors, although we need 
to properly combine them in order to build  efficient
``yes-no'' set  measurements which reconstruct the message bit-by-bit by checking, at each step of the procedure,   whether the received message belongs to one of two possible sets of codewords. 
In a effort to make the paper self-contained, we reproduce a series of known results~\cite{WildeBOOK,wildeSen} providing, in some cases, alternative proofs which are explicitly presented in the framework
which  best fits with the proposed approach.

The paper is organized as follows:
we start in Sec.~\ref{intro1}, where we introduce the notation and state the problem in a rigorous way. 
In Sec.~\ref{mathtool} we present some mathematical tools which are important to derive our results. In particular 
Sec.~\ref{sec:tip} is devoted to review some basic facts about the structure of typical subspaces of a quantum source, while 
Sec.~\ref{lemmas} discusses few Lemmas which allow us to put bounds on the probability of retrieving certain POVM outcomes from states which are close to each other. 
The bisection protocol is introduced in Sec.~\ref{sec:BIS}, identifying a sufficient condition which ensures it can asymptotically attain the Holevo bound in Sec.~\ref{sec:succprobA} 
and presenting three different methods which satisfy this condition. Conclusions are finally given in Sec.~\ref{sec:con}.

\section{The problem: achieving the Holevo bound}\label{intro1} 

Consider a memoryless quantum communication channel described by a completely positive, trace preserving (CPT) mapping~\cite{HolevoBOOK} $\mathcal{T}$
 that Alice (the sender of the communication scheme) uses to transmit classical messages to Bob (the receiver).
Given an alphabet $\cal A$ of classical symbols,  we define  a $N$-element code $\mathcal{C}:=\{\vec{j}^{(1)}, \cdots, \vec{j}^{(N)}\}$ 
as a subset of ${\cal A}^n$ which contains $N$ selected $n$-long strings $\vec{j} := (j_1, \cdots, j_n)$ of elements of $\cal A$: they represent the codewords which are employed by Alice to codify  $N$ distinct classical messages. 
  A quantum encoding  is then realized by assigning a mapping which,  given  $j\in {\cal A}$,  associates to it a density matrix $\sigma_j \in \mathfrak{S}(\mathcal{H})$  of the quantum carrier that propagates 
through the channel. Accordingly each string $\vec{j}\in {\cal A}^n$ will be represented by the product state
 $\sigma_{\vec{j}}:= \sigma_{j_1}\otimes\ldots\otimes\sigma_{j_n}\in \mathfrak{S}(\mathcal{H}^{\otimes n})$, 
 and received by Bob as 
  \begin{eqnarray} \label{outputdens}
\rho_{\vec{j}}: =\rho_{j_1}\otimes\ldots\otimes\rho_{j_n},
  \end{eqnarray}
   where 
 $\rho_{j}:=\mathcal{T}[\sigma_{j}]$ is the output density matrix  corresponding to the input  $\sigma_j$. 
 In this framework  each classical  code $\mathcal{C}$ is associated with a quantum code via the following classical-to-quantum correspondence 
   \begin{eqnarray}
 \mathcal{C} = \{\vec{j}^{(1)}, \cdots, \vec{j}^{(N)}\} \qquad \longrightarrow \qquad   \mathbf{C}:=\{ \rho_{\vec{j}^{(1)}}, \cdots, \rho_{\vec{j}^{(N)}}\}\;,
   \label{CODEQ}
 \end{eqnarray} 
 the  states  $\rho_{\vec{j}^{(\ell)}}$ being those which Bob has to discriminate in order to recover the message  Alice  sent to him while using the code ${\cal C}$.
For such  purpose he will employ a  decoding  POVM  of elements 
   \begin{equation}
   \left\{X_{1},\cdots, X_N, X_0=\mathbf{1}-\sum_{\ell=1}^NX_{\ell}\right\} \;, \label{POVM}
   \end{equation} whose outcome
   represents the inferred value  of the transmitted message. Specifically 
   for $\ell=1, \cdots, N$, the operator $X_\ell$ is associated with the
event where Bob assumes that  the received message is the $\ell$-th one, while $X_0$ is associated with an explicit failure of the decoding stage. 
Accordingly the average error probability  of the quantum code $\mathbf{C}$ can then be computed as
   \begin{equation}
     P_{err}( \mathbf{C}):=\oneover{N}\sum_{\ell =1}^N\left[1- p_{succ}(\ell)\right]= 1- \oneover{N}\sum_{\ell=1}^N p_{succ}(\ell).\label{err}
   \end{equation}
   where 
    \begin{equation}
     p_{succ}(\ell ):=\tr{X_{\ell}\;  \rho_{\vec{j}^{(\ell)}} 
         },\label{succ}
   \end{equation}
   is the 
   probability that Bob will successfully retrieve the $\ell$-th codeword when Alice transmits it.

   In the long message limit $n\rightarrow\infty$, it has been shown~\cite{holevo1,holevo2,schumawest} that $P_{err}( \mathbf{C})$ 
   can be sent to zero if the number of messages scales as $N=2^{nR}$, $R$ 
   being the transmission rate  of the scheme which is bounded by the Holevo theorem. Specifically we must have that 
   \begin{equation} \label{bound}
     R\leq\max_{\{p_j,\sigma_j\}}\chi(\{p_j,\rho_j\})=C_{Hol},
   \end{equation}
   where  on the right-hand-side the maximization is performed over all possible input ensembles $\{p_j,\sigma_j:j\in{\cal A}\}$  obtained by
      selecting the state $\sigma_j$ with probability distribution $p_j$, and where the 
   Holevo information of the associated output ensemble $\{p_j,\rho_j={\cal T}(\sigma_j); j \in {\cal A} \}$ is defined as
   \begin{equation}
     \chi(\{p_j,\rho_j\}):=S\left(\sum_jp_j\rho_j\right)-\sum_jp_jS(\rho_j),\label{holevoinfo}
   \end{equation}
   by means of the Von Neumann entropy $S(\rho)=-Tr[\rho\log_2\rho]$.
   
     It is known that the inequality~(\ref{bound})  is achievable, in the sense that, 
      for any 
   output ensemble ${\cal E} := \{p_j,\rho_j; j \in {\cal A} \}$, one can  identify a set $\mathbf{C}$ of $N\sim2^{n\; \chi(\{p_j,\rho_j\})}$ quantum 
   codewords  and a decoding POVM~(\ref{POVM})  for which the error 
   probability \eqref{err} goes to zero as $n$ increases. 
   This can be done by exploiting what, in classical information theory, is known as Shannon's averaging trick. 
The idea is as follows:  
the  ensemble ${\cal E}$ can be seen as a source which, when operating $n$ times, 
will  produce $n$-long product states $\rho_{\vec{j}}$ 
of the form
\eqref{outputdens} with probability
\begin{equation} \label{jointp}
  p_{\vec{j}}=p_{j_1}p_{j_2}\ldots p_{j_n}.
\end{equation}
Therefore iterating $N$ times this operation,  ${\cal E}$ will be able to 
generate a code $\mathbf{C}$ defined as in Eq.~(\ref{CODEQ})
with probability 
\begin{equation} \label{PROB}
  P(\mathbf{C})=\prod_{\ell=1,\cdots, N} 
  p_{\vec{j}^{(\ell)}}= \prod_{\ell=1,\cdots, N} 
  \prod_{q=1}^{n} p_{j_q^{(\ell)}},
\end{equation}
where $\vec{j}^{(\ell)}$ are the codewords of the classical counterpart ${\cal C}$  of $\mathbf{C}$.
The set  ${\cal S}:= \{ \mathbf{C}, P(\mathbf{C})\}$ defines the statistical collection of the  quantum codes  one can associate to ${\cal E}$
for fixed values of $N$ and $n$.
Accordingly, instead of optimizing  the total error probability~(\ref{err}) of a single element of such a set, we can now  
consider its averaged value  with respect to the probability $P(\mathbf{C})$, 
 \begin{equation}\label{PERRC}
  \ave{P_{err}}_{\cal S}:=\sum_{\mathbf{C}}P(\mathbf{C})P_{err}(\mathbf{C})= 1 -  \oneover{N}\sum_{\ell=1}^N \ave{p_{succ}(\ell)}_{\cal S}  \;,
\end{equation}
the rationale being that if  this quantity  can be forced to  go to zero in the limit $n\rightarrow\infty$ then at least one (actually almost all)  code 
exists in  ${\cal S}$ for which $P_{err}(\mathbf{C})$ tends to zero in the same limit.

 The first proof ~\cite{holevo2,schumawest}  of this fact made use of a single-step decoding POVM~(\ref{POVM}), known as 
  \emph{pretty good measurement} (PGM) or \emph{square root measurement}, which
is extremely efficient from a theoretical point of view but difficult to implement.  
 More recently, a sequential decoding scheme has been introduced~\cite{seq1,seq2,sen}, which makes use of projective ``yes-no'' measurements to verify whether the received 
 state corresponds to a certain codeword or not. Following an arbitrary ordering 
 of codewords, this question is asked for each of them in turn, until either a positive 
 answer is obtained for some $\vec{j}$ or else a negative answer for all the 
 codewords. To some extent the sequential scheme appears to be easier to realize in practice as it decomposes the process into a series of simple steps,  and indeed several 
 proposals have been made for its use in the context of continuos variable communication lines~\cite{HOLWER,cavesDrum,braunstein,weedbrook}.
 Still it has 
 a major drawback in its scaling, since an order of $N=2^{nR}$ operations is 
 required for its application. 
 The protocol presented here is inspired by the sequential decoding but makes use of a bisection method, 
 performing at each step a ``yes-no'' measurement for a set  of possible 
 codewords, whose size is progressively halved, allowing Bob to recover the transmitted message
 bit-by-bit. 

\section{Mathematical tools} \label{mathtool}
This section reviews some basic facts about typical subspaces and presents some inequalities which will be useful in proving the optimality of our
decoding scheme. For a complete description of the following 
properties we refer the reader to Refs.~\onlinecite{WildeBOOK,HolevoBOOK,HOLREV,sen,WINTERPHD}.

\subsection{Typical subspaces}\label{sec:tip} 

Consider the average state \begin{equation}\rho=\sum_{j\in{\cal A}}p_j\rho_j=\sum_x q_x |e_x\rangle\langle e_x|,\end{equation} 
of the quantum  source ${\cal E} := \{p_j,\rho_j; j \in {\cal A} \}$
and its spectral decomposition in terms of the eigenbasis $\{\ket{e_x}\}$ of $\mathcal{H}$ and the eigenvalues
$\{q_x\}$. This induces a classical random variable $X$ with probability 
    distribution $q_x$ which, on $n$ sampling events,  produces the sequence  $\vec{x}=(x_1,\cdots, x_n)$ with probability $q_{\vec{x}} = \prod_{\ell=1}^n q_{x_\ell}$.
    The classical $\delta-$typical subspace $T_\delta^n$ is defined 
    as the subspace of  such sequences whose sample entropy differs from the 
    expected entropy of the random variable for less than a given quantity $\delta>0$:
    \begin{equation}
      T_\delta^n=\left\{\vec{x}:|\bar{H}(\vec{x})-H(X)|\leq\delta\right\},
    \end{equation}
    the sample entropy of a codeword being  
    \begin{equation}
      \bar{H}(\vec{x})=-\oneover{n}\log_2{q_{\vec{x}}}=-\oneover{n}\sum_{i=1}^n\log_2{q_{x_i}},
    \end{equation}
    i.e. the average information content of the $n$ symbols in the $\vec{x}$ 
    sequence, while the associated Shannon entropy is defined as usual:
    \begin{equation}
      H(X)=-\sum_xq_x\log_2q_x = S(\rho),
    \end{equation}
    where in the last identity we used the correspondence with the von Neumann entropy functional of the average state $\rho$.
    As a consequence, the $\delta-$typical subspace $\hil_{typ}^{(n)}$ of quantum state $\rho$ is made of all those vectors $\ket{e_{\vec{x}}}$
    whose corresponding classical sequence is $\delta-$typical, i.e. 
    $\vec{x}\in T_\delta^n$.    
The projector on this subspace is given by
\begin{equation}
  P=\sum_{\vec{x}\in T_\delta^n}\ket{e_{\vec{x}}}\bra{e_{\vec{x}}}.
\end{equation}
Similar properties as for the classical typical subspace hold for the quantum 
one, namely
\begin{align}
  &\tr{P\rho^{\otimes n}}\geq1-\epsilon_1,\label{totaleuno}\\
  &\tr{P}\leq 2^{n\left[S(\rho)+\delta\right]},\label{totaledue}\\
  &2^{-n\left[S(\rho)+\delta\right]}P\leq P\rho^{\otimes 
  n}P\leq 2^{-n\left[S(\rho)-\delta\right]}P,\label{totaletre}
\end{align} 
for $\epsilon_1>0$ and $n$ sufficiently large. These properties state 
respectively that:
\begin{itemize}
  \item The quantum state $\rho^{\otimes n}$ resides with high probability in 
  the $\delta-$typical subspace of $\rho$;
  \item The size of the $\delta-$typical subspace is exponentially smaller than 
  the size of the whole space, unless the source is maximally mixed, i.e. 
  $S(\rho)=\log_2d$;
  \item The probability distribution of $\delta-$typical sequences is approximately uniform
  $\sim 2^{-nS(\rho)}$.
\end{itemize}
It is finally important to observe that the parameter $\epsilon_1$ entering in Eq.~(\ref{totaleuno}) can be linked to $n$ via
an exponential scaling\cite{Gallager}, i.e.  $\epsilon_1 = {O}(e^{-n})$, which ensures that for all polynomial functions $poly(n)$ of $n$  one has 
\begin{eqnarray} \label{polyn}
 \lim_{n\rightarrow \infty} poly({n}) \; \epsilon_1 =0,
\end{eqnarray}  
(see Appendix~\ref{append} for details). 

Similar typical subspaces can be identified also for each specific state $\rho_{\vec{j}}$ 
produced by the source, i.e. for each codeword in $\mathbf{C}$, by using the notion of conditional typicality. 
Indeed each source state can be seen as a classical-quantum 
state $\ket{j}\bra{j}\otimes\rho_j$ and its spectral decomposition will be in 
terms of eigenvectors $\{\ket{j}\otimes\ket{e^j_y}\}$ and eigenvalues 
$\{\lambda^j_y\}$. This again induces the classical random variables $J$, with probability distribution $p_j$ 
representing the possible states emitted by the source, and 
$Y$, with conditional probability distribution $\lambda^j_y=p(y|j)$.
The classical $\delta-$conditionally typical subspace is then defined for each 
$n-$long sequence $\vec{j}$ as 
\begin{equation}
  T_\delta^{\vec{j}}=\left\{\vec{y}:|\bar{H}(\vec{y}|\vec{j})-H(Y|J)|\leq\delta\right\},
\end{equation}
where now the entropic quantities are conditional ones, i.e.
\begin{align}
 & \bar{H}(\vec{y}|\vec{j})=-\oneover{n}\log_2\lambda^{\vec{j}}_{\vec{y}}=-\oneover{n}\sum_{i=1}^n\log_2\lambda^{j_i}_{y_i}\\
 &H(Y|J)=\sum_{j\in\mathcal{A}} p_j H(Y|j)=-\sum_{j,y}p_j\lambda^j_y\log_2\lambda^j_y.
\end{align}
The $\delta-$conditionally typical subspace $\mathcal{H}^{\vec{j}}_{typ}$ of 
quantum codeword state $\rho_{\vec{j}}$ is made of all those vectors $\ket{e^{\vec{j}}_{\vec{y}}}$ 
whose corresponding classical sequence is $\delta-$conditionally typical, i.e. $\vec{y}\in 
T^{\vec{j}}_{\delta}$. 
The projector on this subspace is given by
\begin{equation}
  P_{\vec{j}}=\sum_{\vec{y}\in T^{\vec{j}}_{\delta}}\ket{e^{\vec{j}}_{\vec{y}}}\bra{{e^{\vec{j}}_{\vec{y}}}}.
\end{equation}
Given $\epsilon_2>0$ and $n$ sufficiently large, the following three main properties hold for the conditionally typical subspace:
\begin{eqnarray}
    &&\sum_{\vec{j}} p_{\vec{j}} \; {\tr{P_{\vec{j}}\rho_{\vec{j}}}}\geq1-\epsilon_2,\label{parzialeuno}\\
  &&\sum_{\vec{j}} p_{\vec{j}} \; {\tr{P_{\vec{j}}}}\leq 2^{n\left[\sum_{j\in\mathcal{A}}p_j S(\rho_j)+\delta\right]},\label{parzialedue}\\
  &&2^{-n\left[\sum_{j\in\mathcal{A}}p_jS(\rho_j)+\delta\right]}P_{\vec{j}}\leq P_{\vec{j}}\; \rho_{\vec{j}}\; P_{\vec{j}}\leq 
  2^{-n\left[\sum_{j\in\mathcal{A}}p_j S(\rho_j)-\delta\right]}P_{\vec{j}},\label{parzialetre}
\end{eqnarray}
where in the first two expressions the average~\cite{NOTA} is taken  with respect to the joint probability $p_{\vec{j}}$ of ${\cal E}$ introduced in Eq.~(\ref{jointp}), while the last inequality applies for all $\vec{j}$.
As for  Eq.~(\ref{totaleuno}) we stress that the parameter $\epsilon_2$ of Eq.~(\ref{parzialeuno}) can be chosen to have an exponential scaling in $n$ which guarantees that the condition~(\ref{polyn}) holds also in this case. 
 Note finally that the conditionally typical subspaces of different codewords are in general not orthogonal, 
since they are built using vectors of two spectral decompositions of the same 
space $\mathcal{H}^{\otimes n}$. 

\subsection{Measurement Lemmas}\label{lemmas}
We state here some Lemmas which will be used in the rest of the article. They 
relate in various ways quantum states before and after a measurement, with the slight 
but crucial detail that the latter need not be normalized. Formally one can represent them as \emph{subnormalized} density matrices, i.e. 
positive operators whose trace is smaller than or equal to one.

An explicit  proof of the first three Lemmas can be found in Appendix \ref{appendue}: they refer to properties of the trace norm, which for a generic operator $\theta$, is defined as 
$\norm{\theta}_1=\mbox{Tr}{|\theta|}$ with  $|\theta| = \sqrt{\theta^\dag \theta}$ being the modulus of $\theta$.  The 
last Lemma instead was proved by Sen~\cite{sen} and provides an alternative, useful, way of estimating the 
 error probability of the 
sequential decoding protocol of Refs.~\onlinecite{seq1,seq2}.

 \newtheorem{lemma}{Lemma}
\begin{lemma}\label{appclose}
 (Measurement on approximately close states) Let $\rho,\sigma$ be subnormalized density matrices. Let $E$ be a positive and less-than-one 
 operator, i.e. $0\leq E\leq\mathbf{1}$. Then 
 \begin{equation}
   \tr{E\rho}\geq\tr{E\sigma}-2D(\rho,\sigma),\label{appcloseq}
 \end{equation}
 where $D(\rho,\sigma)=\oneover{2}\norm{\rho-\sigma}_1$ is the trace distance between $\rho$ and $\sigma$. 
\end{lemma}

\begin{lemma}\label{gentop}
  (Gentle operator) Let $\rho$ be a subnormalized density matrix and $E$ a positive and less-than-one operator, i.e. $0\leq 
  E\leq\mathbf{1}$. Let also $\ave{\cdots}$ denote the average with respect to some probability distribution, which $\rho$ and $E$ may depend on. 
  Suppose that, for some $1\geq\epsilon>0$,
  \begin{equation}
    \ave{\tr{E\rho}}\geq 1-\epsilon.\label{gentopequno}
  \end{equation}
  Then
  \begin{equation}
    \ave{D\left(\sqrt{E}\rho\sqrt{E},\rho\right)}\leq \sqrt{\epsilon}.\label{gentopeqdue}
  \end{equation}
\end{lemma}

The two previous lemmas are well known for ordinary density matrices; they can be proved also for subnormalized 
ones by use of the following lemma.
\begin{lemma}\label{trdist}
  (Alternative form of trace norm for subnormalized states) Let $\omega$ be a 
  hermitian operator (in particular, $\omega$ could be a subnormalized density 
  matrix). Then 
  \begin{equation}
    \norm{\omega}_1=\max_{-\mathbf{1}\leq\Lambda\leq\mathbf{1}}\tr{\Lambda\omega}.\label{trdisteq}
  \end{equation}
\end{lemma}

\begin{lemma}\label{contra}
  (Contractivity of trace distance for POVM elements) Let $\rho,\sigma$ be 
  subnormalized density matrices and $0\leq E\leq\mathbf{1}$ a positive and 
  less-than-one operator (for example it could be a POVM element and/or a projector).
  Then
  \begin{equation}
    \trd{E\rho E}{E\sigma E}\leq \trd{\rho}{\sigma}.\label{contreq}
  \end{equation}
\end{lemma}

\begin{proof}
  Consider the expression of the trace norm of a hermitian operator as in Lemma 
  \ref{trdist} and apply it to the LHS of \eqref{contreq}:
  \begin{align}
    2\trd{E\rho E}{E\sigma E}&=\max_{-\mathbf{1}\leq\Lambda\leq\mathbf{1}}\tr{\Lambda 
    E(\rho-\sigma)E}\label{maxuno}\\
    &=\tr{\bar{\Lambda}E(\rho-\sigma)E}=\tr{\Lambda'(\rho-\sigma)}\\
&\leq\max_{-\mathbf{1}\leq\Lambda\leq\mathbf{1}}\tr{\Lambda(\rho-\sigma)}
    =2\trd{\rho}{\sigma}.\label{maxdue}
  \end{align}
  The second equality follows from explicitly using the operator $\bar{\Lambda}$ which 
  attains the maximum in \eqref{maxuno}. The third equality follows from using the cyclic property of the trace and
  setting $\Lambda'=E\bar{\Lambda} E$. The inequality follows from the 
  fact that also $\Lambda'$ is positive and less-than-one. 
\end{proof}

\begin{lemma}\label{Sen} 
(Sen's Lemma) Let $\rho$ be a subnormalized density matrix and  
  $P_1,\ldots,P_k$ orthogonal projectors on 
  subspaces of its Hilbert space. Let also $Q_i=\mathbf{1}-P_i$ be their complementary projectors. 
  Then 
  \begin{equation}
    Tr\left[P_k\ldots P_1\rho P_1\ldots P_k\right]\geq 
    \tr{\rho}-2\sqrt{\sum_{i=1}^k\tr{\rho Q_i}}.
  \end{equation}
  \end{lemma}

\section{The bisection protocol}\label{sec:BIS}
\subsection{Description of the protocol} In this subsection we introduce our decoding protocol (Fig. \ref{figure1}) which, given a density matrix
 extracted from a $N=e^{n R}$-element quantum code $\mathbf{C}$~(\ref{CODEQ}), generated by the source ${\cal E}$, tries to identify it by using a bisection method. The measurement  process comprises of ${u}_{\mbox{\tiny{F}}}= nR$ nested detection events, each aimed to recover one bit of information from the transmitted signal. 
 
As a preliminary step, Bob assigns an ordering of the codewords in $\mathbf{C}$, identifying each of them with
    a unique string of ${u}_{\mbox{\tiny{F}}}$ bits, $\vec{k} = (k_1,  k_2, \cdots, k_{{u}_{\mbox{\tiny{F}}}})$, e.g. by providing a binary representation of their label $\ell\in\{1,\dots,N\}$.  In particular the first bit of the string $\vec{k}$  identifies two distinct subsets of ${\mathbf{C}}$ containing each $N/2$ codewords: the subset 
    ${\mathbf{C}}^{(1)}_0$ formed by the codewords whose corresponding strings start with $k_1=0$, and the subset  ${\mathbf{C}}^{(1)}_{1}$ characterized by those for which instead $k_1=1$. The second bit of the string $\vec{k}$  is then used to  further halve ${\mathbf{C}}^{(1)}_{0}$  and ${\mathbf{C}}^{(1)}_{1}$. Specifically 
for  $k_1=0,1$, ${\mathbf{C}}^{(1)}_{k_1}$ is split into  the sub-subsets ${\mathbf{C}}^{(2)}_{k_1,k_2=0}$ and  ${\mathbf{C}}^{(2)}_{k_1,k_2=1}$  which includes the $N/4$ codewords whose bits strings have $k_1$ as first bit and $k_2=0$ and $k_2=1$ as second bit, respectively.  Proceeding along the same line 
Bob identifies hence a hierarchy  of subsets organized in ${u}_{\mbox{\tiny{F}}}$ sets, 
the $u$-th one being composed  by  $2^u$ disjoint subsets ${\mathbf{C}}^{(u)}_{k_1,k_2,\cdots, k_u}$ labelled by the indexes $k_1$, $k_2$, $\cdots$, $k_u$, and containing each $2^{{u}_{\mbox{\tiny{F}}}-u} = N/2^u$ codewords.
Specifically  ${\mathbf{C}}^{(u)}_{k_1,k_2,\cdots, k_u}$ is the set formed by the codewords whose identifier string $\vec{k}$
admits the value $k_1$ as first bit, the value $k_2$ as second bit, $\cdots$, and the value $k_u$ as the $u$-th bit. By construction    for all $u\in\{ 1, \cdots, {u}_{\mbox{\tiny{F}}}\}$
they  fulfill the identities
\begin{eqnarray} 
&&{\mathbf{C}}^{(u)}_{k_1,k_2,\cdots, k_{u-1},0}\; \bigcap \;   {\mathbf{C}}^{(u)}_{k_1,k_2,\cdots, k_{u-1},1} \; =\; \O \;, \label{inter} \\
&&{\mathbf{C}}^{(u)}_{k_1,k_2,\cdots, k_{u-1},0}\; \bigcup \; {\mathbf{C}}^{(u)}_{k_1,k_2,\cdots, k_{u-1},1}\; = \; {\mathbf{C}}^{(u-1)}_{k_1,k_2,\cdots, k_{u-1}} \;, \label{complete1} 
\end{eqnarray} 
and the completeness relation 
\begin{eqnarray} 
{\mathbf{C}} = \bigcup_{k_1, k_2, \cdots, k_{u} \in \{0,1\}} {\mathbf{C}}^{(u)}_{k_1,k_2,\cdots, k_{u}}\;.
\end{eqnarray} 

To recover which codeword Alice is transmitting,  Bob performs a sequence of ${u}_{\mbox{\tiny{F}}}$ concatenated measurements organized as  shown in 
Fig.~\ref{figure1}. 
The first of these measures  is aimed to determine  the value of the first bit $k_1$ of the bit string associated with  the transmitted codeword, i.e. it allows Bob to determine whether the codeword
is in the subset  ${\mathbf{C}}^{(1)}_0$ or in the subset  ${\mathbf{C}}^{(1)}_1$. The exact form of such procedure will be assigned in the following sections where three alternative examples of the scheme will be discussed in details: for the moment it is sufficient to observe that it can be described as a POVM ${\cal M}^{(1)}$  
 of  elements $N_{0}^{(1)}$, $N_{1}^{(1)}$ associated respectively to the outcomes $k_1=0$ and $k_1=1$, plus a null term $N_{null}^{(1)} = \mathbf{1} -N_{0}^{(1)} -N_{1}^{(1)}$ associated with the case in which no decision can be made on the value of $k_1$: if this event occurs simply Bob declares failure of the decoding procedure and stops the protocol (in the first implementation of the scheme we discuss  in Sec.~\ref{EX0}  this element is not present, which is equivalent to set $N_{null}^{(1)}=0$). 
 Once $k_1$ has been determined, Bob proceeds with the second step of the protocol aimed
 to recover the value of the bit $k_2$ of the transmitted codeword. To this purpose, conditioned on the value of $k_1\in \{0,1\}$ obtained in the previous step, Bob performs now a new POVM ${\cal M}_{k_1}^{(2)}$  aimed to determine whether the received codeword belongs to ${\mathbf{C}}^{(2)}_{k_1,0}$ or to ${\mathbf{C}}^{(2)}_{k_1,1}$.
Also  ${\cal M}_{k_1}^{(2)}$ is  characterized by three elements: $N_{k_1,0}^{(2)}$, $N_{k_1,1}^{(2)}$ corresponding to the cases $k_2=0$ and $k_2=1$ respectively, and  $N_{k_1,null}^{(2)} = \mathbf{1} -N_{k_1,0}^{(2)} -N_{k_1,1}^{(2)}$ corresponding to the failure event (again the
explicit expressions for these operators will be assigned later on). 
  The procedure iterates  till Bob either gets a failure event or recovers all the ${u}_{\mbox{\tiny{F}}}$ bits which identify the transmitted codeword.   Specifically, assuming that no failures have occurred in the first $u-1$ steps yielding the values $k_1$, $k_2$, $\cdots$, $k_{u-1}$ for the associated bits,  at the $u$-th step Bob 
  performs on the system a POVM ${\cal M}_{k_1,k_2, \cdots, k_{u-1}}^{(u)}$ of elements $N_{k_1,k_2, \cdots, k_{u-1}, 0}^{(u)}$, $N_{k_1,k_2, \cdots, k_{u-1}, 1}^{(u)}$, and 
  $N_{k_1,k_2, \cdots, k_{u-1},null}^{(u)} = \mathbf{1} -N_{k_1,k_2, \cdots, k_{u-1}, 0}^{(u)} -N_{k_1,k_2, \cdots, k_{u-1}, 1}^{(u)}$ to decide whether the received codeword belongs
  to set  ${\mathbf{C}}_{k_1,k_2, \cdots, k_{u-1},0}^{(u)}$ or to  ${\mathbf{C}}_{k_1,k_2, \cdots, k_{u-1},1}^{(u)}$.
  
\begin{figure}[h]
	\centering
	\includegraphics[width=.9\textwidth]{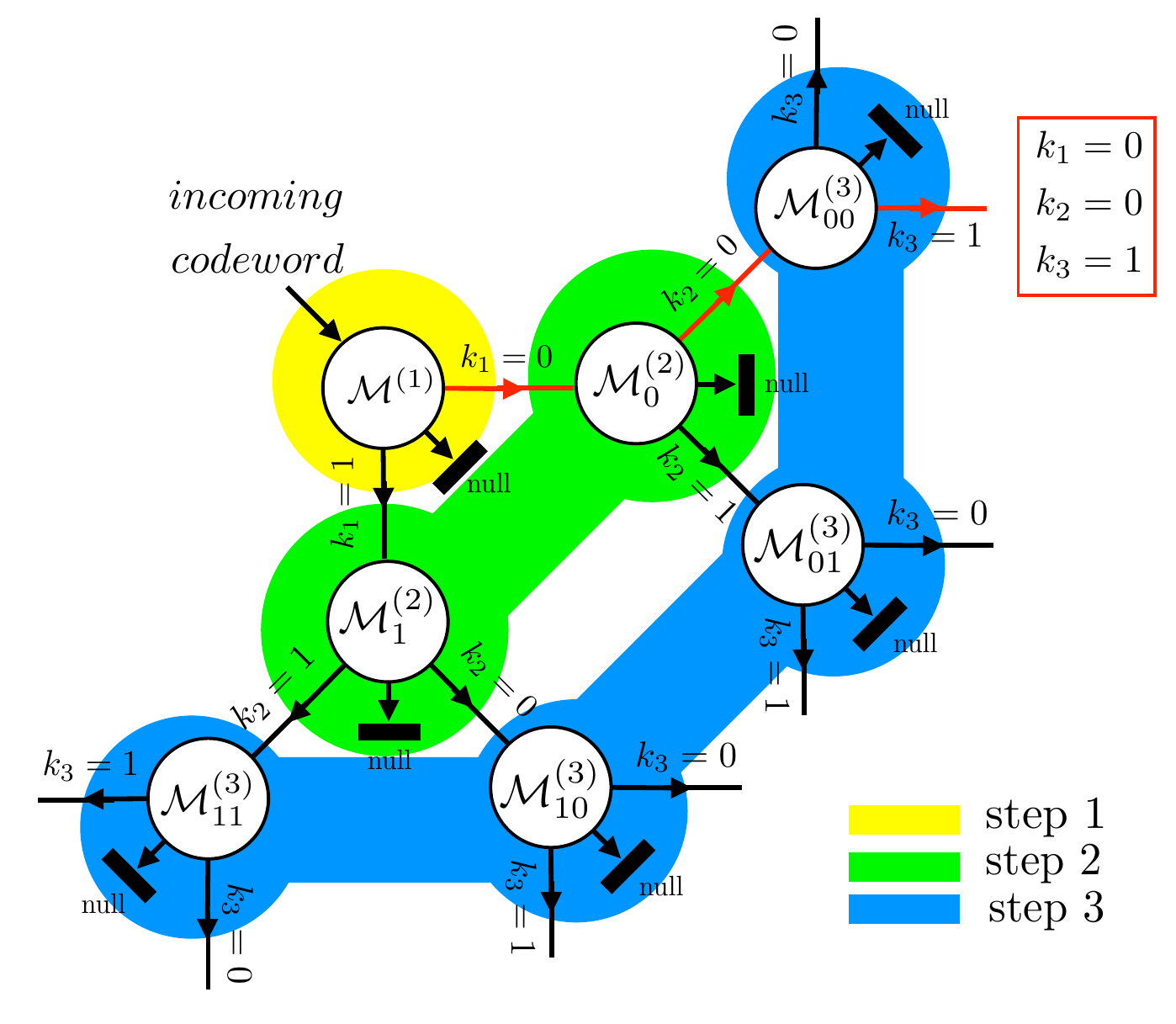}
	\caption{Schematic representation of the bisection decoding procedure. It consists in a sequence of adaptive measurements which are performed in series of 
	 ${u}_{\mbox{\tiny{F}}}$ concatenated steps, each being characterized by a POVM (the white circles) which admits three possible outcomes: two being associated respectively to 
	 the identification of the corresponding bit as  $0$ or $1$, and one, the {\it null} outcome, associated with the event where no decision can be made on the value of the bit.  
	The POVM to be performed at the $u$-th step depends upon the value of the bit obtained at the previous ones: for instance at the step number 2 Bob will perform
	either the POVM ${\cal M}_0^{(2)}$ or the POVM ${\cal M}_1^{(2)}$ depending on the value of $k_1$ he has obtained at the first step of the procedure, while at the step
	number $3$ Bob will perform  the POVMs ${\cal M}_{00}^{(3)}$, ${\cal M}_{01}^{(2)}$, ${\cal M}_{10}^{(3)}$, or ${\cal M}_{11}^{(2)}$ depending on the values of $k_1$ and $k_2$ obtained in the previous two steps.   The figure refers to the case of  ${u}_{\mbox{\tiny{F}}}=3$, the redline representing the trajectory which yields Bob to assign the binary string $\vec{k}=(0,0,1)$ to the received codeword.   }
	\label{figure1}
\end{figure}

 Given the above construction the probability of recovering a given string of bits  $\vec{k}=(k_1,k_2, 
 \cdots, k_{{u}_{\mbox{\tiny{F}}}})$ when measuring an input state $\rho_{\vec{j}}\in \mathbf{C}$ can now be expressed along the lines detailed in Appendix~\ref{conca}, i.e.   \begin{eqnarray} 
  P(\vec{k}| \rho_{\vec{j}}) = \mbox{Tr} [ F_{\vec{k}}\;  \rho_{\vec{j}}  ] \;, \label{fff111} 
  \end{eqnarray} 
  with $F_{\vec{k}}$ defined as 
   \begin{eqnarray}
    F_{\vec{k}}&=&\Big |    \sqrt{\misub{{u}_{\mbox{\tiny{F}}}}{k_1,k_2, \cdots, k_{{u}_{\mbox{\tiny{F}}}}}} \; \sqrt{\misub{{u}_{\mbox{\tiny{F}}}-1}{{k_1,k_2, \cdots, k_{{u}_{\mbox{\tiny{F}}}-1}}}}\cdots \sqrt{\misub{2}{{k}_{1},k_2}}\sqrt{\misub{1}{{k}_{1}}}\Big |^2 \;,
   \label{bisegenerica}
  \end{eqnarray}
  with $\{\misub{{u}}{k_1, \cdots, k_{{u}}}\}_{
u\in\{{1},\cdots,u_{F}=nR\}; k_{1},\cdots,k_{u-1}\in\{0,1\}}$ the operators which define the POVM's of the protocol. 
 The success probability of the procedure 
 follows then from  this 
expression by simply setting $\vec{k}$ to coincide with the binary string associated with the selected $\vec{j}$, e.g.  
in the case of the $\ell$-th codeword 
  \begin{eqnarray}
    {p_{succ}(\ell)}&=&P(\vec{k}^{(\ell)}| \rho_{\vec{j
    }^{(\ell)}}) = \mbox{Tr} [ F_{\vec{k}^{(\ell)}} \rho_{\vec{j
    }^{(\ell)}} ] ,
    \label{psucc}
  \end{eqnarray}
  where $\vec{k}^{(\ell)}$ is the binary string corresponding to the index $\ell$ which defines the selected vector~$\vec{j
    }^{(\ell)}$.

 \subsection{Attainability of the Holevo Bound via a bisection protocol}  \label{sec:succprobA}

 As detailed in the previous paragraphs a bisection protocol aimed to decode a $N=2^{nR}$-codewords quantum code  $\mathbf{C}$  is defined by assigning a family of three-outcomes POVM's $\mathcal{M}^{(u)}_{k_{1},\cdots,k_{u-1}}$  identified by the integer index $u\in\{{1},\cdots,u_{F}=nR\}$  and by the binary labels  $k_{1},\cdots,k_{u-1}\in\{0,1\}$ and concatenated as schematically shown in Fig.~\ref{figure1}.
 In this section we are going to show that, as long as the rate $R$ respects the inequality~(\ref{bound}), 
 it is possible to assign the operators $\{\misub{{u}}{k_1, \cdots, k_{{u}}}\}_{
u\in\{{1},\cdots,u_{F}=nR\}; k_{1},\cdots,k_{u-1}\in\{0,1\}}$
 which define the  measures $\{ \mathcal{M}^{(u)}_{k_{1},\cdots,k_{u-1}}\}_{
u\in\{{1},\cdots,u_{F}=nR\}; k_{1},\cdots,k_{u-1}\in\{0,1\}}
 $ in such a way that, in the limit of large $n$, the corresponding success probability~(\ref{psucc}) converges asymptotically to 1 when averaged over all possible codes generated by the output source ${\cal E}$, i.e. 
   \begin{eqnarray}\label{dopolemma12323}
 \lim_{n \rightarrow \infty} \langle p_{succ}(\ell)\rangle_{\cal S} =1 \;.
    \end{eqnarray}
   Accordingly the corresponding average error probability~(\ref{PERRC}) asymptotically nullifies, 
   proving hence that bisection decoding procedures can be used to  saturate the Holevo bound. 
In order to achieve this goal we start by presenting a sufficient condition  on $N^{(u)}_{k_{1},\cdots,k_{u}}$ which, if fulfilled, would yield the limit~(\ref{dopolemma12323}) independently of the value of $R$, see Theorem~\ref{MainTh}. Subsequently we show that  for all
rates $R$  respecting the Holevo Bound~(\ref{bound}) we can indeed fulfil such sufficient condition. This is done  by presenting three independent choices of  the operators $\{\misub{{u}}{k_1, \cdots, k_{{u}}}\}_{
u\in\{{1},\cdots,u_{F}=nR\}; k_{1},\cdots,k_{u-1}\in\{0,1\}}$, corresponding to three different ways of constructing the bisection scheme: via orthogonal projections (see Sec.~\ref{EX0}); via PGM detections (see Sec.~\ref{EX00}); and via sequential detections (see Sec.~\ref{EX1}). 
\newtheorem{Main}{Theorem}
\begin{Main}[Sufficient Condition] \label{MainTh}
For  $n$ integer let $\mathbf{C}$ be a  quantum code formed by $N=2^{nR}$ separable codewords~(\ref{outputdens}) of length $n$ extracted from the output ensemble $\mathcal{E}$,  
and a bisection protocol with POVM's $\{ \mathcal{M}^{(u)}_{k_{1},\cdots,k_{u-1}}\}_{
u\in\{{1},\cdots,u_{F}=nR\}; k_{1},\cdots,k_{u-1}\in\{0,1\}}
 $ characterized by the
operators $\{\misub{{u}}{k_1, \cdots, k_{{u}}}\}_{
u\in\{{1},\cdots,u_{F}=nR\}; k_{1},\cdots,k_{u-1}\in\{0,1\}}$.
The corresponding success probability~(\ref{psucc}) 
 converges to one as in Eq.~(\ref{dopolemma12323}) when averaged over all possible codes generated by the output ensemble ${\cal E}$
 if, for all $\ell \in \{ 1, \cdots, N\}$ and $u\in \{ 1, \cdots, u_F\}$,  one of the following conditions is fulfilled
\item{i)} 
\begin{align}
\ave{\operatorname{Tr}\left[\misub{{u}}{k_1^{(\ell)}, \cdots, k^{(\ell)}_{{u}}}\rho_{\vec{j}^{(\ell)}}\right]}_{\mathcal{S}}\geq 1-\epsilon(n);\label{IMPOIMPO}
\end{align}
\item{ii)} 
\begin{align}
\ave{\operatorname{Tr}\left[\misub{{u}}{k_1^{(\ell)}, \cdots, k^{(\ell)}_{{u}}} P \rho_{\vec{j}^{(\ell)}}P\right]}_{\mathcal{S}}\geq 1-\epsilon(n),\label{IMPOIMPO11}
\end{align}
with $\epsilon(n)>0$ being a function that decreases asymptotically to zero faster than $1/n^2$ as $n$ goes to infinity and $P$ being the projector on the typical subspace of the average codeword associated with the output ensemble ${\cal E}$ 
(in the above expression $k_1^{(\ell)}, \cdots, k^{(\ell)}_{{u}}$ are the first $u$ elements of the binary string $\vec{k}^{(\ell)}$ which represents the codeword index $\ell$ that labels the density matrix $\rho_{\vec{j}^{(\ell)}}$).
\end{Main} 

\begin{proof} 
We start by directly proving  that Eq.~(\ref{IMPOIMPO}) is a sufficient condition for Eq.~(\ref{dopolemma12323}), part i) of the theorem.  Part ii) of the theorem is then obtained by showing that 
Eq.~(\ref{IMPOIMPO11}) implies Eq.~(\ref{IMPOIMPO}). 

{\it Part i)}: The success probability~(\ref{succ}) that an element $\rho_{\vec{j}^{(\ell)}} \in {\mathbf C}$ will be correctly decoded by the bisection procedure characterized by the operators $\{\misub{{u}}{k_1, \cdots, k_{{u}}}\}_{
u\in\{{1},\cdots,u_{F}=nR\}; k_{1},\cdots,k_{u-1}\in\{0,1\}}$ can be computed as in Eq.~(\ref{psucc}) 
with  $\vec{k}^{(\ell)}$ being the identifying bit string that Bob has assigned to the $\ell$-th codeword.
  To put a bound on the average value of this quantity over the collection ${\cal S}$ of quantum codes emitted by the source ${\cal E}$,  we 
 observe that
  \begin{eqnarray}
    \Big\langle p_{succ}(\ell)\Big\rangle_{\cal S}
     &=& 
    \ave{\tr{    M^2_{{u}_{\mbox{\tiny{F}}}}  \; M_{{u}_{\mbox{\tiny{F}}}-1}
  \; \cdots  \; M_{{1}}\; 
    \rho_{\vec{j}^{(\ell)}} \; M_{{1}}\; 
 \cdots\;M_{{u}_{\mbox{\tiny{F}}}-1}
    }}_{\cal S}  \nonumber \\ &\geq&\ave{\tr{   M^2_{{u}_{\mbox{\tiny{F}}}}    \rho_{\vec{j}^{(\ell)}}}-2\trd{M_{{u}_{\mbox{\tiny{F}}}-1}\ldots M_1\rho_{\vec{j}^{(\ell)} }M_1\ldots 
M_{{u}_{\mbox{\tiny{F}}}-1}}{\rho_{\vec{j}^{(\ell)}}}}_{\cal S}, \label{EQU1}
    \end{eqnarray}
    where  
for easiness of notation we introduced $M_u=\sqrt{\misub{{u}}{k_1^{(\ell)}, \cdots, k^{(\ell)}_{{u}}}}$ and applied Lemma \ref{appclose} with 
  $E=M^2_{{u}_{\mbox{\tiny{F}}}} $,
 $\rho=  M_{{u}_{\mbox{\tiny{F}}}-1}
  \; \cdots  \; M_{{1}}\; 
    \rho_{\vec{j}^{(\ell)}} \; M_{{1}}\; 
 \cdots\;M_{{u}_{\mbox{\tiny{F}}}-1}$, 
    and $\sigma= \rho_{\vec{j}^{(\ell)}}$.
        By use of the triangular inequality we also observe that 
    \begin{eqnarray}
  \label{ufmenouno}   
  && \trd{M_{{u}_{\mbox{\tiny{F}}}-1}\ldots M_1\rho_{\vec{j}^{(\ell)} }M_1\ldots 
M_{{u}_{\mbox{\tiny{F}}}-1}}{\rho_{\vec{j}^{(\ell)} }}
\nonumber\\ && \qquad  \leq \trd{M_{{u}_{\mbox{\tiny{F}}}-1}\rho_{\vec{j}^{(\ell)} }M_{{u}_{\mbox{\tiny{F}}}-1}}{\rho_{\vec{j}^{(\ell)} }} + 
\trd{M_{{u}_{\mbox{\tiny{F}}}-1}\ldots M_1\rho_{\vec{j}^{(\ell)} }M_1\ldots 
M_{{u}_{\mbox{\tiny{F}}}-1}}{M_{{u}_{\mbox{\tiny{F}}}-1}\rho_{\vec{j}^{(\ell)} }M_{{u}_{\mbox{\tiny{F}}}-1}} \nonumber \\
&&\qquad  \leq \trd{M_{{u}_{\mbox{\tiny{F}}}-1}\rho_{\vec{j}^{(\ell)} }M_{{u}_{\mbox{\tiny{F}}}-1}}{\rho_{\vec{j}^{(\ell)} }}
+ 
\trd{M_{{u}_{\mbox{\tiny{F}}}-2}\ldots M_1\rho_{\vec{j}^{(\ell)} }M_1\ldots 
M_{{u}_{\mbox{\tiny{F}}}-2}}{\rho_{\vec{j}^{(\ell)} }} \nonumber \\
&&\qquad  \leq 
\sum_{u=1}^{{u}_{\mbox{\tiny{F}}}-1}\trd{M_u\rho_{\vec{j}^{(\ell)}}M_u}{\rho_{\vec{j}^{(\ell)}}},
    \end{eqnarray}
 where the second inequality follows from Lemma \ref{contra} while the third one by direct iteration of the previous passages.  
Replaced into Eq.~(\ref{EQU1})   this finally yields 
    \begin{align}\label{dopolemma}
 \Big\langle p_{succ}(\ell)\Big\rangle_{\cal S}\geq\ave{\tr{   M^2_{{u}_{\mbox{\tiny{F}}}}      \rho_{\vec{j}^{(\ell)}}}}_{\cal S}
 -2 \sum_{u=1}^{{u}_{\mbox{\tiny{F}}}-1}\ave{\trd{M_u\rho_{\vec{j}^{(\ell)}}M_u}{\rho_{\vec{j}^{(\ell)}}}}_{\cal S} .
    \end{align}
Assume now that Eq.~(\ref{IMPOIMPO})  holds. Accordingly  for all $u\in\{ 1,\cdots, {u}_{\mbox{\tiny{F}}}\}$, $\ell=1,\cdots,N$ we have 
\begin{eqnarray} 
\ave{\tr{M_u^2{\rho}_{\vec{j}^{(\ell)}}}}_{\cal S}\geq 1- \epsilon(n), 
\end{eqnarray}
with $\epsilon(n)$ being a positive function which goes to zero faster than $1/n^2$. 
Then thanks to Lemma~\ref{gentop} we can write 
  \begin{align}\label{dopolemma1}
 \Big\langle p_{succ}(\ell)\Big\rangle_{\cal S}\geq 1- \epsilon(n) -2 n R \sqrt{\epsilon(n)} ,
    \end{align}
which forces  $\Big\langle p_{succ}(\ell)\Big\rangle_{\cal S}$  to converge to 1 as $n\rightarrow \infty$. This shows that   Eq.~(\ref{IMPOIMPO}) is indeed a  sufficient condition for ~(\ref{dopolemma12323}). 

{\it Part ii):} To prove  that  Eq.~(\ref{IMPOIMPO11}) is a sufficient condition for~(\ref{dopolemma12323}) we invoke 
 Lemma \ref{appclose} with $E=M_{u}^{2}$, $\rho=\rho_{\vec{j}^{(\ell)}}$ 
  and $\sigma=P {\rho}_{\vec{j}^{(\ell)}}P$ obtaining the following inequality 
  \begin{eqnarray} \label{misuraSmoothed}
 \ave{\tr{M_u^2{\rho}_{\vec{j}^{(\ell)}}}}_{\cal S}\geq\ave{\tr{M_u^2 P {\rho}_{\vec{j}^{(\ell)}}P}}_{\cal S} -{2D\left(\rho_{\vec{j}^{(\ell)}},P{\rho}_{\vec{j}^{(\ell)}}P\right)}.
  \end{eqnarray}
From   Eq.~\eqref{totaleuno} we also know that 
        for $n$ sufficiently large and $\epsilon_1=O(e^{-n})$ one has 
     \begin{eqnarray}
       \ave{\tr{P\rho_{\vec{j}^{(\ell)} }}}_{\cal S}=\tr{P\ave{\rho_{\vec{j}^{(\ell)} }}_{\cal S}}&=&\tr{P\rho^{\otimes n}}\geq 1-\epsilon_1,\label{genttotaluno}
     \end{eqnarray}
     where we used the fact that the average over ${\cal S}$ of the $\ell$-th codeword corresponds to the average with respect to the joint probability~(\ref{jointp}) of 
     $\rho_{\vec{j}}$, i.e. 
     \begin{equation}
       \ave{\rho_{\vec{j}^{(\ell)}}}_{\cal S} =  \sum_{\vec{j}}p_{\vec{j}}\rho_{\vec{j}}=\rho^{\otimes 
       n}.
     \end{equation}
     Accordingly via 
Lemma \ref{gentop} we can conclude that 
\begin{eqnarray} 
 \ave{D\left(\rho_{\vec{j}^{(\ell)}},P{\rho}_{\vec{j}^{(\ell)}}P\right)}_{\cal S}\leq\sqrt{\epsilon_1},\label{genttotaldue}
     \end{eqnarray}
which inserted in  Eq.~(\ref{misuraSmoothed}) finally yields  
  \begin{eqnarray}
    \ave{\tr{M_u^2{\rho}_{\vec{j}^{(\ell)}}}}_{\cal S}\geq\ave{\tr{M_u^2 P {\rho}_{\vec{j}^{(\ell)}}P}}_{\cal S} -2\sqrt{\epsilon_1},
   \label{smoonosmoo}
  \end{eqnarray}
which thanks to~(\ref{IMPOIMPO11}) implies the inequality~(\ref{IMPOIMPO}) and hence, via part i)  of the theorem, Eq.~(\ref{dopolemma12323}).
\end{proof}

 Thanks to Theorem~\ref{MainTh} we can now prove that bisection protocols allows one to  attain the Holevo bound, by 
 showing that, for all rates $R$ respecting~(\ref{bound}), it is possible to identify operators $\{\misub{{u}}{k_1, \cdots, k_{{u}}}\}_{
u\in\{{1},\cdots,u_{F}=nR\}; k_{1},\cdots,k_{u-1}\in\{0,1\}}$ fulfilling Eq.~(\ref{IMPOIMPO11}). 
 Ideally one way of building the POVMs  ${\cal M}_{k_1,k_2,\cdots, k_{u-1}}^{(u)}$ which define the bisection decoding procedure, would be to
  identify its elements $N_{k_1,k_2, \cdots, k_{u-1}, 0}^{(u)}$, $N_{k_1,k_2, \cdots, k_{u-1}, 1}^{(u)}$ with  the projectors on the subspaces spanned by the codewords of sets ${\mathbf C}^{(u)}_{k_1,k_2,\cdots,k_{u-1}, 0}$ and ${\mathbf C}^{(u)}_{k_1,k_2,\cdots,k_{u-1}, 1}$
    respectively. This is not possible however due to the fact that such spaces are in general not orthogonal, though we expect typical subspaces of different codewords of the source 
  to be disjoint in the long $n$ limit: some kind of regularization is hence necessary. 
 In the following we shall present three alternative, yet asymptotically equivalent,  ways to realize this:  
the first makes use of orthogonal projections on subspaces identified by treating  asymmetrically the set ${\mathbf C}^{(u)}_{k_1,k_2,\cdots,k_{u-1}, 0}$ and ${\mathbf C}^{(u)}_{k_1,k_2,\cdots,k_{u-1}, 1}$, 
the second is based on the PGM construction, and finally the third makes use of the POVM elements of the sequential protocol of Refs.~\onlinecite{seq1,seq2,sen}. 

 \subsubsection{Method  1: orthogonal projections}\label{EX0}
Consider the set ${\mathbf C}^{(u)}_{k_1,k_2,\cdots,k_{u-1}, 0}$. For each  one of its codewords $\rho_{\vec{j}^{(\ell)}}$   we can associate a typical subspace 
$\mathcal{H}^{\vec{j}^{(\ell)}}_{typ}$ and a corresponding projector $P_{\vec{j}^{(\ell)}}$ 
along the lines detailed in Sec.~\ref{sec:tip}. Next we construct the subspace ${\cal H}^{(u)}_{k_1,k_2,\cdots,k_{u-1}, 0}$ spanned by the vectors  which 
can be written as a direct sum of the elements of the $\mathcal{H}^{\vec{j}^{(\ell)}}_{typ}$s of ${\mathbf C}^{(u)}_{k_1,k_2,\cdots,k_{u-1}, 0}$, i.e. 
\begin{eqnarray} 
{\cal H}^{(u)}_{k_1,k_2,\cdots,k_{u-1}, 0} := \bigoplus_{\ell \in {\mathbf C}^{(u)}_{k_1,k_2,\cdots,k_{u-1}, 0}} \mathcal{H}^{\vec{j}^{(\ell)}}_{typ}
\end{eqnarray}  
where the sum  is performed over the $\ell$s whose corresponding  vector $\rho_{\vec{j}^{(\ell)}}$ belongs to the set   ${\mathbf C}^{(u)}_{k_1,k_2,\cdots, k_{u-1}, 0}$.
 By construction it follows that each one of the $\mathcal{H}^{\vec{j}^{(\ell)}}_{typ}$ associated to  ${\mathbf C}^{(u)}_{k_1,k_2,\cdots,k_{u-1}, 0}$ are proper
 subspaces of ${\cal H}^{(u)}_{k_1,k_2,\cdots,k_{u-1}, 0}$. Accordingly, indicating with $P^{(u)}_{k_1,k_2,\cdots,k_{u-1}, 0}$ the projector on ${\cal H}^{(u)}_{k_1,k_2,\cdots,k_{u-1}, 0}$
 we have that 
 \begin{eqnarray} \label{NEWIMPO} 
 P^{(u)}_{k_1,k_2,\cdots,k_{u-1}, 0} \geq P_{\vec{j}^{(\ell)}},
 \end{eqnarray} 
for all $\ell \in {\mathbf C}^{(u)}_{k_1,k_2,\cdots,k_{u-1}, 0}$. Also due to the partial overlapping among the $\mathcal{H}^{\vec{j}^{(\ell)}}_{typ}$ of ${\mathbf C}^{(u)}_{k_1,k_2,\cdots,k_{u-1}, 0}$
the sum of  the associated $P_{\vec{j}^{(\ell)}}$s  will in general be larger than $P^{(u)}_{k_1,k_2,\cdots,k_{u-1}, 0}$, i.e. 
\begin{eqnarray} 
P^{(u)}_{k_1,k_2,\cdots,k_{u-1}, 0} \leq  \sum_{\ell \in {\mathbf C}^{(u)}_{k_1,k_2,\cdots,k_{u-1}, 0}}  P_{\vec{j}^{(\ell)}}. \label{NEWIMPO1} 
 \end{eqnarray} 
We define the orthogonal projections method for the bisection POVM by identifying $N_{k_1,k_2, \cdots, k_{u-1}, 0}^{(u)}$ with $P^{(u)}_{k_1,k_2,\cdots,k_{u-1}, 0}$ and 
$N_{k_1,k_2, \cdots, k_{u-1}, 1}^{(u)}$ with its complementary counterpart, i.e.
\begin{eqnarray}
N_{k_1,k_2, \cdots, k_{u-1}, 0}^{(u)} &:=& P^{(u)}_{k_1,k_2,\cdots,k_{u-1}, 0} ,\label{orthOne} \\
N_{k_1,k_2, \cdots, k_{u-1}, 1}^{(u)} &:=&  Q^{(u)}_{k_1,k_2,\cdots,k_{u-1}, 0} = \mathbf 1 - P^{(u)}_{k_1,k_2,\cdots,k_{u-1}, 0}. \label{orthTwo}
\end{eqnarray} 
A couple of remarks are mandatory:
\begin{itemize}
\item[i)] notice that  $N_{k_1,k_2, \cdots, k_{u-1}, 1}^{(u)}$ does  not coincide with the projector $P^{(u)}_{k_1,k_2,\cdots,k_{u-1}, 1}$ on the subspace ${\cal H}^{(u)}_{k_1,k_2,\cdots,k_{u-1}, 1}$
formed by the direct sum of the typical subspaces   $\mathcal{H}^{\vec{j}^{(\ell)}}_{typ}$ associated with ${\mathbf C}^{(u)}_{k_1,k_2,\cdots,k_{u-1}, 1}$.
Notice also that, due to the partial overlapping of the typical subspaces of different codewords, in general we can neither  establish
 an inequality similar to Eq.~(\ref{NEWIMPO}) which links  $N_{k_1,k_2, \cdots, k_{u-1}, 1}^{(u)}$  and the $ P_{\vec{j}^{(\ell)}}$ of ${\cal H}^{(u)}_{k_1,k_2,\cdots,k_{u-1}, 1}$,
nor fix an ordering between $N_{k_1,k_2, \cdots, k_{u-1}, 1}^{(u)}$ and 
$P^{(u)}_{k_1,k_2,\cdots,k_{u-1}, 1}$;
\item[ii)] by construction the scheme we are analyzing here does not include the possibility of the null event described in the previous section. Indeed in this case we have 
\begin{eqnarray} 
N_{k_1,k_2, \cdots, k_{u-1}, null}^{(u)}= \mathbf{1} - N_{k_1,k_2, \cdots, k_{u-1}, 0}^{(u)}-N_{k_1,k_2, \cdots, k_{u-1}, 1}^{(u)}=0.
\end{eqnarray} 
The associated set  POVM ${\cal M}_{k_1,k_2,\cdots, k_{u-1}}^{(u)}$ is thus a projective measurement which admits only two possible outcomes, $k_u=0$  and $k_u=1$. 
\end{itemize}

From Theorem~\ref{MainTh} 
 the asymptotic attainability of the Holevo bound with this procedure  can be established by 
   showing that  Eq.~(\ref{IMPOIMPO11}) holds, i.e. explicitly 
\begin{lemma}
For all rates $R$ satisfying the Holevo bound~(\ref{bound})
the bisection scheme associated with operators $\{\misub{{u}}{k_1, \cdots, k_{{u}}}\}$
defined as in Eqs.~(\ref{orthOne}), (\ref{orthTwo}) fullfils the sufficient condition~\textnormal{(\ref{IMPOIMPO11})}.
\end{lemma}  
\begin{proof}
   Consider first the case with $k_u=0$. Given then a generic codeword $\rho_{\vec{j}^{(\ell)}}$ of  ${\mathbf C}^{(u)}_{k_1,k_2,\cdots, k_{u-1}, 0}$
 we can write 
 \begin{eqnarray}  \label{IMPOIMPO111}
&& \ave{\tr{N_{k_1^{(\ell)},k_2^{(\ell)}, \cdots, k_{u-1}^{(\ell)}, 0}^{(u)}\;  \bar{\rho}_{\vec{j}^{(\ell)}}}}_{\cal S}=\ave{\tr{P_{k_1^{(\ell)},k_2^{(\ell)}, \cdots, k_{u-1}^{(\ell)}, 0}^{(u)}\;  \bar{\rho}_{\vec{j}^{(\ell)}}}}_{\cal S} \geq \ave{\tr{P_{\vec{j}^{(\ell)}}
\;  \bar{\rho}_{\vec{j}^{(\ell)}}}}_{\cal S} \nonumber \\
&&= \sum_{\vec{j}} p_{\vec{j}} \;  {\tr{P_{\vec{j}}\bar{\rho}_{\vec{j}}}} \geq 
 \sum_{\vec{j}} p_{\vec{j}} \;\Big(  {\tr{P_{\vec{j} }\rho_{\vec{j}}}} -2 {\trd{\bar{\rho}_{\vec{j}}}{\rho_{\vec{j}}}}\Big) 
      \geq 1 - \epsilon_2 - 2\sqrt{\epsilon_1},
\end{eqnarray}
where for easy of notation we set  $\bar{\rho}_{\vec{j}^{(\ell)}}:= P {\rho}_{\vec{j}^{(\ell)}}P$  and where 
 we used the fact that taking the average with respect to the statistical collection ${\cal S}$ of $\tr{P_{\vec{j}^{(\ell)}}
\;  \bar{\rho}_{\vec{j}^{(\ell)}}}$ is equivalent to taking the average of  
${\tr{P_{\vec{j}}\bar{\rho}_{\vec{j}}}}$ with respect to $p_{\vec{j}}$, i.e. 
\begin{eqnarray} \label{prima} 
\ave{\tr{P_{\vec{j}^{(\ell)}}
\;  \bar{\rho}_{\vec{j}^{(\ell)}}}}_{\cal S} = \sum_{\vec{j}} p_{\vec{j}} \;  {\tr{P_{\vec{j}}\bar{\rho}_{\vec{j}}}}\;.
\end{eqnarray} 
The first inequality of Eq.~(\ref{IMPOIMPO111}) follows  from Eq.~(\ref{NEWIMPO}); 
the second  inequality follows instead from applying Lemma \ref{appclose} with 
    $E=P_{\vec{j}}$, $\rho=\bar{\rho}_{\vec{j}}$ and $\sigma=\rho_{\vec{j}}$; while finally 
the third inequality     follows  both from the high probability of projecting 
    codeword $\rho_{\vec{j}}$ on its conditionally typical subspace \eqref{parzialeuno}
    and from the same concept for the average codeword, together with Lemma \ref{gentop} (as in (\ref{genttotaluno},\ref{genttotaldue})),
    the parameters $\epsilon_1$ and $\epsilon_2$ being both exponentially small in $n$ to guarantee the limit property~(\ref{polyn}). 
    Equation~(\ref{IMPOIMPO111}) proves hence that Eq.~(\ref{IMPOIMPO11}) applies at least for the sets  ${\mathbf C}^{(u)}_{k_1,k_2,\cdots, k_u}$  with $k_u=0$. 

  Take next   $k_u=1$ and  a generic codeword $\rho_{\vec{j}^{(\ell)}}$ of  ${\mathbf C}^{(u)}_{k_1,k_2,\cdots, k_{u-1}, 1}$. In this case we have 
 \begin{eqnarray}  \label{IMPOIMPO112pre}
  \ave{\tr{N_{k_1^{(\ell)},k_2^{(\ell)}, \cdots, k_{u-1}^{(\ell)}, 1}^{(u)}\;  \bar{\rho}_{\vec{j}^{(\ell)}}}}_{\cal S}&=& \ave{\tr{ \bar{\rho}_{\vec{j}^{(\ell)}}}}_{\cal S}  -\ave{\tr{P_{k_1^{(\ell)},k_2^{(\ell)}, \cdots, k_{u-1}^{(\ell)}, 0}^{(u)}\;  \bar{\rho}_{\vec{j}^{(\ell)}}}}_{\cal S}
\nonumber \\
&\geq& \ave{\tr{ \bar{\rho}_{\vec{j}^{(\ell)}}}}_{\cal S}  - \sum_{\ell'\neq \ell}    \ave{\tr{P_{\vec{j}^{(\ell')}}\;  \bar{\rho}_{\vec{j}^{(\ell)}}}}_{\cal S}
\end{eqnarray}
where the inequality follows from~(\ref{NEWIMPO1}) plus  adding all the remaining  terms $P_{\vec{j}^{(\ell')}}$ associated with codewords having $\ell'\neq\ell$.
Observe then that from Eq.~(\ref{totaleuno}) we have 
\begin{eqnarray} \label{prima1} 
\ave{\tr{ \bar{\rho}_{\vec{j}^{(\ell)}}}}_{\cal S}  = \sum_{\vec{j}} p_{\vec{j}} \tr{ P {\rho}_{\vec{j}}} = \tr{P \rho^{\otimes n} } \geq 1 -\epsilon_1\;,
\end{eqnarray} 
with $\epsilon_1$ being an exponentially small 	function of $n$. Furthermore  for each term of the sum on the RHS of Eq.~(\ref{IMPOIMPO112}) we have
    \begin{align} \ave{\tr{  P_{\vec{j}^{(\ell')} } \bar{\rho}_{\vec{j}^{(\ell)}} }}_{\cal S}= \sum_{\vec{j}, \vec{j'} } 
       p_{\vec{j}}\;  p_{\vec{j}'} \;{\tr{\bar{\rho}_{\vec{j}} P_{\vec{j}'}}} &= \sum_{\vec{j}'} p_{\vec{j}'}\; \tr{\overline{\rho^{\otimes n}} 
      {P_{\vec{j}'}}}\leq \norm{\overline{\rho^{\otimes n}}}_\infty 
     \sum_{\vec{j}'} p_{\vec{j}'} \; {\tr{P_{\vec{j}'}}} \nonumber \\
      &\leq 
      2^{-n[S(\rho)-\delta]} \; 2^{n\left[\sum_{j}p_{j}S(\rho_j)+\delta\right]}=2^{-n[\chi(\{p_j, \rho_j\})-2\delta]},\label{pdiversoj1}
    \end{align}
    where the second inequality follows from typical subspaces' properties (\ref{totaletre},\ref{parzialedue}) and where     $\chi(\{p_j, \rho_j\})$ is  Holevo information \eqref{holevoinfo} of the source ${\cal E}$. 
  Replacing~(\ref{prima1}) and (\ref{pdiversoj1}) into Eq.~(\ref{IMPOIMPO112pre}) we arrive hence to 
   \begin{eqnarray}  \label{IMPOIMPO112}
  \ave{\tr{N_{k_1^{(\ell)},k_2^{(\ell)}, \cdots, k_{u-1}^{(\ell)}, 1}^{(u)}\;  \bar{\rho}_{\vec{j}^{(\ell)}}}}_{\cal S}
&\geq& 1 -\epsilon_1 - 2^{n R} \; 2^{-n[\chi(\{p_j, \rho_j\})-2\delta]} \;, 
\end{eqnarray}
which shows that as long as  the rate  $R$ respects the  Holevo bound~(\ref{bound}), i.e. 
  \begin{eqnarray} 
     R<\chi(\{p_j, \rho_j\})-2\delta,  \label{IMPO0} 
     \end{eqnarray} 
    for some $\delta >0$,  Eq.~(\ref{IMPOIMPO11}) applies also for  the  sets  ${\mathbf C}^{(u)}_{k_1,k_2,\cdots, k_u}$ with $k_u=1$.  
\end{proof}
The inequalities~(\ref{IMPOIMPO111}) and (\ref{IMPOIMPO112}) prove that under the constraint~(\ref{IMPO0}) the proposed implementation of the bisection decoding scheme asymptotically attains the Holevo bound, yielding
an average error probability which converges to zero in the limit of $n\rightarrow \infty$. 

 \subsubsection{Method 2: via PGM detections} \label{EX00}
   An alternative way to implement the bisection 
    protocol is substituting the sequential set  measurement $N$ with one
    inspired by the Pretty Good Measurement (PGM), first introduced to 
    demonstrate the achievability of the Holevo bound. 
   
   For each set   ${\mathbf C}^{(u)}_{{k_1}, \cdots, k_{u}}$ define the positive operator 
    \begin{equation}
      S^{(u)}_{{k_1},\cdots, k_{u}}=\sum_{\ell'\in{\mathbf C}^{(u)}_{k_1,\cdots, k_{u}}} P_{\vec{j}^{(\ell')}}, 
    \end{equation}
    i.e. the sum of projectors of all the codewords in that set . From the 
    non-orthogonality of projectors and the completeness property \eqref{complete1} 
    it follows 
    \begin{equation}
      S^{(u-1)}_{{k_1},\cdots,k_{u-1}}=S^{(u)}_{{k_1},\cdots,k_{u-1},0}+S^{(u)}_{{k_1},\cdots,k_{u-1},1}\geq\mathbf{1}.
    \end{equation}
    Thus we can build the $u-$th measurement to decide whether the word belongs to $\mathbf{C}^{(u)}_{{k_1},\cdots, k_{u-1}, 0}$ or $\mathbf{C}^{(u)}_{{k_1},\cdots, k_{u-1}, 1}$
    by using the sum operators for these two sets, renormalized 
    by the sum operator for $\mathbf{C}^{(u-1)}_{{k_1},\cdots, k_{u-1}}$, which contains both of them at the previous step:
    \begin{align}\label{PGMdets}
\misub{u}{k_1,\cdots, k_{u}}=\left[{S^{(u-1)}_{k_1,\cdots, k_{u-1}}}\right]^{-1/2} S^{(u)}_{{k_1,\cdots, k_{u}
}}\left[{ 
S^{(u-1)}_{{k_1,\cdots, k_{u-1}}
}}\right]^{-1/2},
       \end{align}
       where the inverse  $\left[{S^{(u-1)}_{k_1,\cdots, k_{u-1}}}\right]^{-1/2}$ is meant to be computed only on the support of ${S^{(u-1)}_{k_1,\cdots, k_{u-1}}}$ (otherwise 
       the  operator is assumed to be null). 
    In this way we obtain a proper set  POVM, since the renormalization allows us to 
    take into account the intersections between typical subspaces of different 
    codewords, i.e. 
    \begin{eqnarray}
0\leq   \misub{u}{k_1,\cdots, k_u} &\leq& \sum_{k\in 0,1}  \misub{u}{k_1,\cdots, k} =\left[{S^{(u-1)}_{k_1,\cdots, k_{u-1}}}\right]^{-1/2} \left[ S^{(u)}_{{k_1,\cdots,0}}
+ S^{(u)}_{{k_1,\cdots,1}}\right] \left[{ 
S^{(u-1)}_{{k_1,\cdots, k_{u-1}}
}}\right]^{-1/2}\nonumber 
\leq {\mathbf 1} \;.
 \label{impo11111} 
  \end{eqnarray}
As in the case discussed previously we can now show that for all $R$ fulfilling the Holevo bound~(\ref{bound}) the operators~(\ref{PGMdets}) 
satisfy the sufficient condition Eq.~(\ref{IMPOIMPO11}) of Theorem~\ref{MainTh}. 

\begin{lemma}
For all rates $R$ satisfying the Holevo bound~(\ref{bound})
the bisection scheme associated with operators $\{\misub{{u}}{k_1, \cdots, k_{{u}}}\}$
defined as in Eq.~(\ref{PGMdets}) fullfils the sufficient condition~\textnormal{(\ref{IMPOIMPO11})}.
%
\end{lemma}  
\begin{proof}
Observe that 
  \begin{eqnarray}
    \ave{\tr{ {\misub{{u}}{k_1^{(\ell)}, \cdots, k^{(\ell)}_{{u}}}} \; 
    \bar{\rho}_{\vec{j}^{(\ell)}}
    }}_{\cal S} 
      &\geq& \ave{\tr{ \left[{ 
S^{(u-1)}_{{k_1^{(\ell)},\cdots, k_{u-1}^{(\ell)}}
}}\right]^{-1/2}\; P_{\vec{j}^{(\ell)}}  \left[{ 
S^{(u-1)}_{{k_1^{(\ell)},\cdots, k_{u-1}^{(\ell)}}
}}\right]^{-1/2}
    \bar{\rho}_{\vec{j}^{(\ell)}}   }}_{\cal S}\nonumber \\
    &=& \ave{\tr{\Lambda_{\ell} \bar{\rho}_{\vec{j}^{(\ell)}} }}_{\cal S}=\ave{p_{succ}^{(u-1)}(\ell)}_{\cal S},\label{PGMred}
    \end{eqnarray} 
    where the latter is the average success probability of recovering the $\ell$-th codeword 
        from the set  $\mathbf{C}^{(u-1)}_{k_1^{(\ell)},\cdots, k_{u-1}^{(\ell)}}$ 
    while using a PGM strategy and $\Lambda_{\ell}$ is the corresponding POVM 
    element.
    Accordingly  we can bound each of the terms on the RHS of Eq.~(\ref{dopolemma}) by exploiting the efficiency of the PGM protocol. 
       Specifically, we employ the Hayashi-Nagaoka inequality\cite{hayanaga} 
       \begin{equation}
         \mathbf{1}-\Lambda_{\ell}\leq 2 Q_{\vec{j}^{(\ell)}} + 4 \sum_{\ell'\neq\ell} 
         P_{\vec{j}^{(\ell')}}
       \end{equation}
 to write the average success probability as 
 \begin{align}
   \ave{p_{succ}^{(u-1)}(\ell)}_{\cal S} &\geq \ave{\tr{\bar{\rho}_{\vec{j}^{(\ell)}} }}_{\cal S} - 2\ave{\tr{Q_{\vec{j}^{(\ell)}} \bar{\rho}_{\vec{j}^{(\ell)}} } 
   }_{\cal S} - 4\sum_{\ell'\neq\ell}\ave{P_{\vec{j}^{(\ell')}} \bar{\rho}_{\vec{j}^{(\ell)}} }_{\cal 
   S}\\
   &\geq 1- \epsilon_1 - 2 (\epsilon_2+2\sqrt{\epsilon_1}) - 4\cdot 2^{nR}\cdot2^{-n[\chi(\{p_j, 
   \rho_j\})-2\delta]},
 \end{align}
 where the last inequality follows from Eqs.~(\ref{genttotaluno}) and (\ref{pdiversoj1}) 
 and the fact that
 \begin{eqnarray}
 \ave{\tr{\bar{\rho}_{\vec{j}^{(\ell)} } 
      Q_{\vec{j}^{(\ell)} }}}_{\cal S}&=&\ave{\tr{\bar{\rho}_{\vec{j}^{(\ell)}}}}_{\cal S} -\ave{\tr{P_{\vec{j}^{(\ell)} }\bar{\rho}_{\vec{j}^{(\ell)} }}}_{\cal S}\nonumber \\
      &\leq& 1 -\ave{\tr{P_{\vec{j}^{(\ell)} }\bar{\rho}_{\vec{j}^{(\ell)} }}}_{\cal S}  =1-  \sum_{\vec{j}} p_{\vec{j}} \;  {\tr{P_{\vec{j}}\bar{\rho}_{\vec{j}}}}
    \nonumber \\
      &\leq & 1-  \sum_{\vec{j}} p_{\vec{j}} \;\Big(  {\tr{P_{\vec{j} }\rho_{\vec{j}}}} -2 {\trd{\bar{\rho}_{\vec{j}}}{\rho_{\vec{j}}}}\Big) 
      \leq \epsilon_2+ 2\sqrt{\epsilon_1},\label{qugualej}
    \end{eqnarray}
  which is derived as in Eq.\eqref{IMPOIMPO111}.
 Similarly to what we observed in Eq.~(\ref{IMPOIMPO112}) it then follows that if the rate $R$ fulfills the constraint~(\ref{IMPO0}) 
    for some $\delta >0$, then  for  $n$ sufficiently large one has that, for all $u$ and $\ell$, 
    \begin{equation}
      \ave{\tr{ {\misub{{u}}{k_1^{(\ell)}, \cdots, k^{(\ell)}_{{u}}}} \; 
    \bar{\rho}_{\vec{j}^{(\ell)}}}}_{\cal S}  \geq   \ave{\tr{\Lambda_{\ell} \bar{\rho}_{\vec{j}^{(\ell)}} }}_{\cal S}
     \geq  1-\epsilon_3,\label{last}
    \end{equation}
    with $\epsilon_3=O(e^{-n})$ being exponentially small in $n$ and fulfilling the condition~(\ref{polyn}) showing hence that 
     (\ref{IMPOIMPO11}) is satisfied by the selected operators.  
\end{proof}

  \subsubsection{Method 3: via sequential POVM}\label{EX1}
 Another way to    regularize set-projection operators necessary to implement the bisection scheme, is to 
 make use of  the sequential protocol for that 
  set , but without gaining knowledge about the result of this subroutine. Accordingly the regularized set-projection operators will be implemented as a black box, applying the sequential decoding scheme to the set of 
codewords which appear inside that set , taking also into account failure in projecting on the typical subspace of previous codewords, in the code ordering chosen by Bob, see Fig.~\ref{figure2}. The resulting
setting is clearly redundant as the vast majority of information gathered via the sequential decoding is simply neglected in the process.  Also, the same
procedure is iterated every time a  new bit of the bijective encoding has to be acquired, increasing hence the chances of deteriorating   the transmitted codeword.
Still, as we shall see in the following, the scheme is efficient enough to allow for the saturation of the Holevo bound. 
\begin{figure}[th]
	\centering
	\includegraphics[width=.9\textwidth]{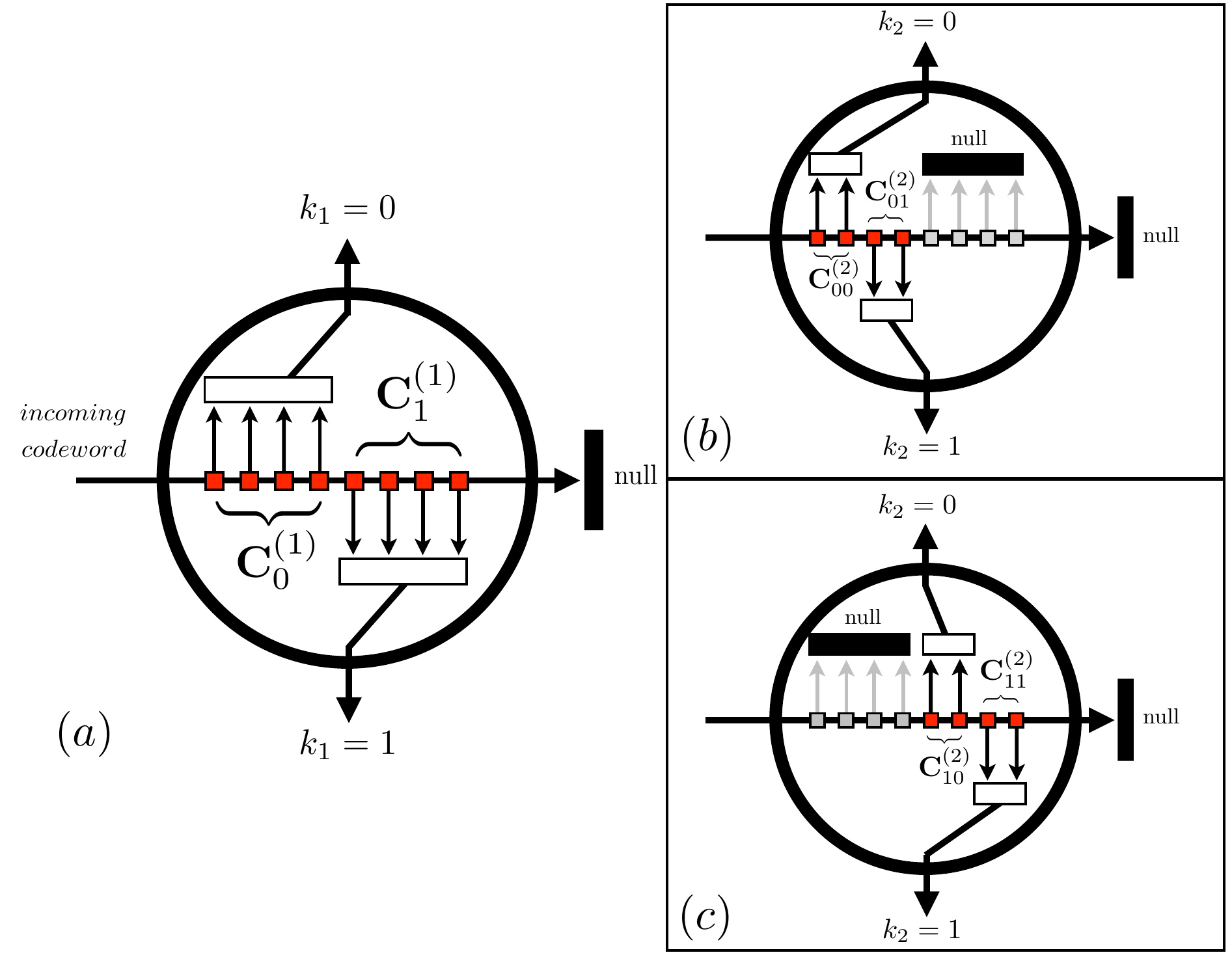}
	\caption{Schematic representation of the set  POVMs ${\cal M}^{(1)}$, ${\cal M}_0^{(2)}$, and ${\cal M}_1^{(2)}$  in terms of the sequential POVM decoding procedure (little square elements) for $N=8$ codewords. Panel (a): implementation of ${\cal M}^{(1)}$. The red color of the square blocks indicates that all the  elements of the sequential decoding POVM are active: their outcomes are used to determine whether the
	incoming codeword belongs to the subset  ${\mathbf  C}_0^{(1)}$ (first four codewords), or to the subset  ${\mathbf C}_1^{(1)}$ (last four codewords) fixing the value of $k_1$. The rectangular elements of the figure indicate that no other information is extracted from the outcomes of the sequential measurement. Panel (b):  implementation of ${\cal M}_0^{(2)}$ which discriminates
	between the subsets ${\mathbf C}_{00}^{(2)}$ and ${\mathbf C}_{01}^{(2)}$.
	This element operates on the state emerging from the port $k_1=0$ of ${\cal M}^{(1)}$, see e.g. Fig.~\ref{figure1}. As indicated by the color, only the  first elements of the sequential POVM are active, while the outputs of the remaining ones are equivalent to the null result. Panel (c): implementation of ${\cal M}_1^{(2)}$ which discriminates among  ${\mathbf C}_{10}^{(2)}$ and ${\mathbf C}_{11}^{(2)}$. }
	\label{figure2}
\end{figure}
  \\

  In order to formalize this construction, for each quantum  codeword $\rho_{\vec{j}^{(\ell)}}\in {\mathbf C}$, we write its corresponding element of the 
  sequential POVM~\cite{seq1,seq2,sen} as
  \begin{eqnarray}
  E_{1}&=&P_{\vec{j}^{(1)}} \;,\nonumber \\
    E_{\ell} &=&Q_{\vec{j}^{(1)}}\ldots Q_{\vec{j}^{(\ell-1)}}P_{\vec{j}^{(\ell)}}Q_{\vec{j}^{(\ell-1)}}\ldots 
    Q_{\vec{j}^{(1)}}, \qquad \ell \geq 2\;,
  \end{eqnarray}
  where $P_{\vec{j}}$ is the projector on the typical subspace of codeword state $\rho_{\vec{j}}$ and $
  Q_{\vec{j}}=\mathbf{1}-P_{\vec{j}}$ its complementary.
  We remind that by construction these operators fulfill the proper normalization condition, 
  \begin{eqnarray} 
 &&0\leq   E_\ell \leq \mathbf{1} \;,\\
 && 0\leq \sum_{\ell = 1}^{N} E_\ell = \mathbf{1} -E_0 \leq \mathbf{1} \;, 
 \end{eqnarray} 
and that, given a  density matrix   $\rho_{\vec{j}}\in \mathbf{C}$,  the probability of recovering the codeword  $\vec{j}^{(\ell)}$ is given  by 
\begin{eqnarray} \label{seqprob} 
P_{seq} (\vec{j}^{(\ell)}| \rho_{\vec{j}}) = \mbox{Tr} [E_\ell\rho_{\vec{j}} ]\;. 
\end{eqnarray}
Using this expression we can hence estimate the probability that $\rho_{\vec{j}^{(\ell)}}$
belongs to the set  ${\mathbf C}^{(u)}_{k_1,k_2,\cdots, k_u}$ by simply summing the above expression over all $\vec{j}^{(\ell)}$  belonging to such set , i.e. 
\begin{eqnarray} 
P(\rho_{\vec{j}^{(\ell)}} \in {\mathbf C}^{(u)}_{k_1,k_2,\cdots, k_u}) = \sum_{\ell\in{\mathbf C}^{(u)}_{k_1,k_2,\cdots, k_u}} P_{seq} (\vec{j}^{(\ell)}| \rho_{\vec{j}})  = \mbox{Tr} [\misub{u}{k_1,k_2,\cdots, k_u}\rho_{\vec{j}} ]\;,
 \end{eqnarray} 
  where the sum is performed over the $\ell$s whose corresponding  vector $\rho_{\vec{j}^{(\ell)}}$ belongs to the set   ${\mathbf C}^{(u)}_{k_1,k_2,\cdots, k_u}$,
and where  \begin{equation}\label{SeqPOVM}
    \misub{u}{k_1,k_2,\cdots, k_u}=\sum_{\ell\in{\mathbf C}^{(u)}_{k_1,k_2,\cdots, k_u}} E_\ell,
  \end{equation}
  is the   set-sequential-measurement
 associated with the set  ${\mathbf C}^{(u)}_{k_1,k_2,\cdots, k_u}$ induced by the sequential decoding POVM.
Since for all $u\in\{ 1,\cdots, {u_F}\}$ and for all $k_1$, $k_2$, $\cdots$, $k_{u-1}$, the sets  ${\mathbf C}^{(u)}_{k_1,k_2,\cdots, k_{u-1},0}$ and ${\mathbf C}^{(u)}_{k_1,k_2,\cdots, k_{u-1},1}$ are not overlapping (see e.g. Eq.~(\ref{inter}))  we have that 
\begin{eqnarray}
0\leq   \misub{u}{k_1,k_2,\cdots, k_u} \leq \sum_{k\in 0,1}  \misub{u}{k_1,k_2,\cdots, k} \leq 
\sum_{\ell=1}^N E_\ell \leq \mathbf{1}\;, \label{impo} 
  \end{eqnarray}
 which  guarantees that the operators 
    $\misub{u}{k_1,k_2,\cdots, k_{u-1},0}$,  $\misub{u}{k_1,k_2,\cdots, k_{u-1},1}$, and  $\misub{u}{k_1,k_2,\cdots, k_{u-1},null} = \mathbf{1}-\sum_{k\in 0,1}  \misub{u}{k_1,k_2,\cdots, k}$ form a properly normalized POVM.


   \begin{lemma}
   For all rates $R$ satisfying the Holevo bound~(\ref{bound})
the bisection scheme associated with operators $\{\misub{{u}}{k_1, \cdots, k_{{u}}}\}$
defined as in Eq.~(\ref{SeqPOVM}) fullfils the sufficient condition~\textnormal{(\ref{IMPOIMPO11})}.
\end{lemma}  
   \begin{proof}
   First observe  that each operator ${\misub{{u}}{k_1^{(\ell)}, \cdots, k^{(\ell)}_{{u}}}}$ 
    is the sum of a certain number of sequential POVM elements, 
    always containing the element $E_{\ell}$ corresponding to 
    the right codeword. Since all the operators in the sum are positive we can state 
    \begin{align}
    \ave{\tr{ {\misub{{u}}{k_1^{(\ell)}, \cdots, k^{(\ell)}_{{u}}}} \; 
    \bar{\rho}_{\vec{j}^{(\ell)}}}}_{\cal S} 
      =\sum_{{\ell'}\in{\mathbf C}^{(u)}_{k_1^{(\ell)},\cdots, k_u^{(\ell)}}} 
 \ave{\tr{E_{\ell'}  \; \bar{\rho}_{\vec{j}^{(\ell)}}}}_{\cal S}
      \geq \ave{\tr{E_{\ell} 
      \bar{\rho}_{\vec{j}^{(\ell)}}}}_{\cal S} , \label{IMPO1}
    \end{align}
    where in the last term we recognize the average success probability (\ref{seqprob}) of the sequential protocol
    computed on the  subnormalized version $\bar{\rho}_{\vec{j}^{(\ell)}}$ of the $\ell$-th codeword. 
Accordingly 
       we can bound each of the terms on the RHS of Eq.~(\ref{dopolemma}) by exploiting the efficiency of the sequential protocol.
      Specifically by applying Sen's Lemma \ref{Sen} and using the concavity of the square root function  we can 
    write:
    \begin{align}
    \ave{\tr{E_{\ell} 
      \bar{\rho}_{\vec{j}^{(\ell)}}}}_{\cal S}&=\ave{\tr{P_{\vec{j}^{(\ell)}}Q_{\vec{j}^{(\ell -1)}}\ldots Q_{\vec{j}^{(1)}}\; \bar{\rho}_{\vec{j}^{(\ell)}}\; Q_{\vec{j}^{(1)}}\ldots Q_{\vec{j}^{(\ell-1)}} 
      P_{\vec{j}^{(\ell)}}}}_{\cal S}\\&\geq\ave{\tr{\bar{\rho}_{\vec{j}^{(\ell)} }}}_{\cal S}-2\ave{\sqrt{{\tr{\bar{\rho}_{\vec{j}^{(\ell)}} Q_{\vec{j}^{(\ell)}}}} 
      +\sum_{\ell' \neq \ell }   {\tr{\bar{\rho}_{\vec{j}^{(\ell)}} P_{\vec{j}^{(\ell')}}}} }}_{\cal S} \nonumber \\ 
      &\geq\ave{\tr{\bar{\rho}_{\vec{j}^{(\ell)} }}}_{\cal S}-2\sqrt{\ave{\tr{\bar{\rho}_{\vec{j}^{(\ell)}} Q_{\vec{j}^{(\ell)}}}}_{\cal S} 
      +\sum_{\ell' \neq \ell }   \ave{\tr{\bar{\rho}_{\vec{j}^{(\ell)}} P_{\vec{j}^{(\ell')}}}}_{\cal S} },
    \end{align}
    having added under square root all the terms $P_{\vec{j}^{(\ell')}}$ with $\ell'>\ell$.
    The term outside the square-root can be treated as in \eqref{genttotaluno}.
    For the first term under square-root simply apply Eq. \eqref{qugualej}. 
    For each term of the sum under square-root we can instead use the inequality~(\ref{pdiversoj1}). 
   Therefore we can write 
      \begin{equation}
      \ave{\tr{E_{\ell} 
      \bar{\rho}_{\vec{j}^{(\ell)}}}}_{\cal S} \geq 
      1-\epsilon_1-2\sqrt{\epsilon_2+2\sqrt{\epsilon_1}+2^{nR}\cdot2^{-n[\chi(\{p_j, \rho_j\})-2\delta]}},
    \end{equation}
  which, via Eq.~(\ref{IMPO1}) implies again that for rates $R$ fulfilling Eq.~(\ref{IMPO0}) 
    for some $\delta >0$, then  for  $n$ sufficiently large one has that, for all $u$ and $\ell$, 
    \begin{equation}
      \ave{\tr{ {\misub{{u}}{k_1^{(\ell)}, \cdots, k^{(\ell)}_{{u}}}} \; 
    \bar{\rho}_{\vec{j}^{(\ell)}}}}_{\cal S}  \geq  \ave{\tr{E_{\ell} 
      \bar{\rho}_{\vec{j}^{(\ell)}}}}_{\cal S}  \geq  1-\epsilon_3,\label{last}
    \end{equation}
    with $\epsilon_3=O(e^{-n})$ being exponentially small in $n$ and fulfilling the condition~(\ref{polyn}).
       This proves~(\ref{IMPOIMPO11}) and hence the asymptotic achievability of the Holevo bound with the bisection protocol defined by operators~(\ref{SeqPOVM}). 
    \end{proof}

\section{Conclusions} \label{sec:con}
    In this article we computed an upper bound for the average error probability (over all 
    codewords in a code and over all possible codes) of the bisection decoding 
    scheme. The bound is shown to approach zero exponentially fast  with the 
    codewords' length, for any output ensemble $\mathcal{E}$ whose size is strictly 
    less than $2^{\chi(\mathcal{E})}$. Thus we provided a new proof of the 
    attainability of the Holevo bound for classical communication through a quantum 
    channel for a class of decoding schemes based on the bisection 
    method, whose complexity scales as the logarithm of the codewords' length. 
   An advantage of this protocol is the possibility of gaining a bit of information at each step of the procedure, unlike the 
    sequential decoding, which gives either full or null information about the 
    codeword at each step. This is particularly powerful in the case of failure at a certain step of the protocol, allowing the receiver to 
   at least make use of the previous steps for a partial identification of
    the message. Note also that there is a certain degree of freedom in the implementation of the specific sets' ``yes-no'' 
    measurements, which form a complete POVM at each step, independently of the rest of the 
    protocol, as long as their average error probability approaches zero 
    faster than $n^{-2}$ (e.g. exponentially decaying) as the codewords' length $n$ grows, for all sources respecting the Holevo bound. 
    This fact has been shown by providing three different POVMs which satisfy 
    the bound, employing projectors on typical subspaces and 
    renormalizing for their non-orthogonality. Unfortunately, as in the case of polar coding, this general requirements on the ``yes-no'' set  measurements do not allow one to evaluate their implementation complexity, thus leaving open the problem of determining a structured device for the efficient detection of long codewords.
    Eventually we stress the importance of the Chernoff bound to provide an 
    exponential scaling to the small quantities used in describing the typical 
    subspaces' properties, which in turn allows the convergence of the decoding 
    scheme.

  \appendix

    \section{The law of large numbers via Chernoff bound}\label{append}
    In this appendix we compute an exponential bound for the law of large 
    numbers, which guarantees the convergence of the error probability of our protocol to zero.
   Indeed consider the small quantities $\epsilon_1$, $\epsilon_2$ 
    which appear in Sec.~\ref{sec:tip}. These quantities 
    describe the high probability of finding respectively the average state $\rho^{\otimes n}$ and the codeword states $\rho_{\vec{j}}$
    in their typical subspaces, identified by the projectors $P$ and $P_{\vec{j}}$.
    This is why they are connected, through the classical typical subspaces, to 
    the law of large numbers. \\ 
    Consider for example the average state $\rho$ of the 
    source. We can easily prove that the probability of $n$ copies of the quantum state $\rho$ 
are in its $\delta-$typical subspace, $\tr{P\rho^{\otimes n}}$, is equivalent to the probability of a random 
sample sequence $\vec{x}$ of the corresponding classical source 
being in the classical $\delta-$typical subspace, $Pr(\vec{x}\in T_\delta^n)$:
\begin{align}
    \tr{P\rho^{\otimes n}}&=\tr{\sum_{\vec{x}\in T_\delta^n}\ket{e_{\vec{x}}}\bra{e_{\vec{x}}} \sum_{\vec{x}'}q_{\vec{x}'}\ket{e_{\vec{x}'}}\bra{e_{\vec{x}'} 
    }}\\
    &=\sum_{\vec{x}\in T_\delta^n} 
    \sum_{\vec{x}'}|\braket{e_{\vec{x}}}{e_{\vec{x}'}}|^2q_{\vec{x}'}\\&=\sum_{\vec{x}\in 
    T_\delta^n}q_{\vec{x}}=Pr(\vec{x}\in T_\delta^n).
\end{align}  
A similar result is obtained for each codeword state $\rho_{\vec{j}}$, 
namely
\begin{equation}
  \tr{P_{\vec{j}}\rho_{\vec{j}}}=Pr\left(\vec{y}\in T_\delta^{\vec{j}}\right).
\end{equation}
These probabilities can be bounded from above with the help of the law of large 
numbers. Consider for example the average typical subspace and choose the random variable $Z$, taking values $z=-\log_2{q_{x}}$.
 We also choose the same probability distribution both for $X$ and $Z$, i.e. $q_z\equiv q_x$. Then the law of large numbers states 
that, for any $\delta>0$, the probability that the average of $Z$ over $n$ 
extractions,
\begin{equation}
  \oneover{n}\sum_{i=1}^n z_i=\bar{H}(\vec{x}),
\end{equation} i.e. the sum of $n$ i.i.d. random variables, differs from its
expected value,
\begin{equation}
  \sum_{i=1}^n q_{z_i}z_i=H(X),
\end{equation}
for more than $\delta$ is lower than a small and positive quantity 
$1\gg\epsilon>0$, i.e.
\begin{equation}
  Pr(\vec{x}\in
  T_\delta^n)=Pr\left(|\bar{H}(\vec{x})-H(X)|\geq\delta\right)\leq\epsilon.
\end{equation}
In usual derivations of this result the Chebyshev inequality is exploited, which 
gives a scaling behaviour $\epsilon\sim n^{-1}$. This is not sufficient for 
convergence of the error probability to $0$ for long sequences $n\rightarrow\infty$ 
in \eqref{dopolemma1}. Recalling also that the Chebyshev bound gives a dependence on the variance of the distribution, it is clear that such a scaling is a rough extimate, since the 
law of large numbers is known to be valid also for infinite-variance 
distributions. We therefore use the Chernoff bound to obtain a faster, 
indeed exponential, convergence.\\
Consider first the Markov inequality, valid for any nonnegative random variable 
$t>0$ and $\delta>0$:
\begin{align}
  Pr(t\geq\delta)&=\sum_{t\geq\delta}p_t\leq\sum_{t\geq\delta}p_t\frac{t}{\delta}\\
  &\leq\oneover{\delta}\sum_t t p_t=\frac{\bar{t}}{\delta},
\end{align}
where we used a bar sign to indicate the average over the probability 
distribution of the random variable. The first inequality follows from 
introducing terms which certainly are less than one, given the constraint on the 
sum. The second inequality follows from adding positive terms to the sum, since
the random variable is positive. We now choose $t=e^{sw}$, with $w$ a new random 
variable\footnote{Note that the Chebyshev inequality can be obtained by choosing instead $t=(w-\bar{w})^2$.}
and $\delta=e^{sA}$, without loss of generality.
The Markov inequality then reads
\begin{equation}
  Pr(e^{sw}\geq e^{sA})\leq e^{-sA}g_w(s)
\end{equation}
for any $s,A$, where we called $g_w(s)=\overline{\exp{(sw)}}$ the moment generating function of the random variable 
$w$, i.e.
\begin{equation}
  \overline{w^n}=\frac{d^n g_w(s)}{ds^n}\Big|_{s=0}.
\end{equation}
Now observe that the above inequality between exponentials has two 
different meanings depending on the sign of $s$, implying both
\begin{align}
  &Pr(w\geq A)\leq e^{-sA}g_w(s)\quad s>0\label{codapiu}\\
  &Pr(w\leq A)\leq e^{-sA}g_w(s)\quad s<0.\label{codameno}
\end{align}
These two relations give bounds on the tails of the $w$ probability 
distribution. In order to evaluate how tight such bounds are, we consider the specific case of $w$ being the sum of $n$ i.i.d random variables $x_i$, 
implying for the moment generating function
\begin{align}
  g_w(s)=\overline{\exp\left(s\sum_{i=1}^n 
  x_i\right)}=\prod_{i=1}^n\overline{e^{sx_i}}=\left(g_x(s)\right)^n,
\end{align} 
and take $A=na$, without loss of generality. The previous inequalities become
\begin{align}
  &Pr\left(\oneover{n}\sum_{i=1}^nx_i\geq a\right)\leq \exp\left[-n\left(sa-\ln 
  g_x(s)\right)\right]\quad s>0\label{cp}\\
  &Pr\left(\oneover{n}\sum_{i=1}^nx_i\leq a\right)\leq \exp\left[-n\left(sa-\ln 
  g_x(s)\right)\right]\quad s<0\label{cm}.
\end{align}
We now need to evaluate the behaviour of the coefficient function in the 
exponential: 
\begin{equation}
  h(s)=sa-\ln g_x(s).
\end{equation}
Consider first some properties of $\mu_x(s)=\ln g_x(s)$, following from the nature of the 
moment generating function:
\begin{itemize}
  \item $\mu_x(s=0)=0$, since $g_x(s=0)=1$;
  \item $\mu'_x(s=0)=g'_x(s=0)/g_x(s=0)=\bar{x}$, since $g_x(s=0)=\bar{x}$;
  \item it is convex \begin{align}
  \mu''_x(s)&=\frac{g''_x(s)}{g_x(s)}-\left(\frac{g'_x(s)}{g_x(s)}\right)^2\\
  &=\ave{x^2}_{e}-\ave{x}_{e}^2=\ave{\left(x-\ave{x}_e\right)^2}_e\geq 0
  \quad \forall s,
\end{align}
where we have indicated with \begin{equation}
\ave{f(x)}_e=\frac{\overline{f(x)e^{sx}}}{\overline{e^{sx}}}
\end{equation} 
the probability average with weight $e^{sx}$.
\end{itemize}
 From the previous properties, it follows that the slope of the function, 
starting at $\bar{x}$ at the origin, increases for $s>0$ and decreases for 
$s<0$. Expanding $\mu_x(s)$ for small $s$ at second order, we have for the 
coefficient function
\begin{equation}
  h(s)\simeq (a-\bar{x})s-\frac{s^2}{2}\mu''_x(0).\label{approxcoeff}
\end{equation}
This approximate function (for small $s$) is zero at
\begin{equation}
  s^*=\frac{2(a-\bar{x})}{\mu''_x(0)}.\label{zeropoint}
\end{equation}
Consider now the $s>0$ inequality \eqref{cp}. If $a>\bar{x}$, then the zero $s^*$ is 
positive and inside the range of validity of the inequality. Thus $h(s)>0$ for all $s<s^*$ in the range: the first inequality has a tight bound. Vice versa if $a<\bar{x}$, the zero $s^*$ 
is negative and $h(s)<0$ in the whole range of validity of the first inequality, making it useless.
\\ The situation is reversed when considering the $s<0$ inequality \eqref{cm}. In this case we 
need $s^*<0$, i.e. $a<\bar{x}$, and for any $s>s^*$ in the range the coefficient function will be positive 
again, providing a tight bound for the second inequality.
By calling
\begin{equation}
  h_p=\sup_{s>0}h(s),\quad h_m=\sup_{s<0}h(s),\qquad h_p,h_m>0, 
\end{equation}
the supremum of $h(s)$ in each region, we can thus rewrite the inequalities 
as tight bounds, taking respectively $a=\bar{x}+\delta>\bar{x}$ in the first inequality and $a=\bar{x}-\delta<\bar{x}$, 
in the second one, for any $\delta>0$:
\begin{align}
  &Pr\left(\oneover{n}\sum_{i=1}^nx_i-\bar{x}\geq \delta\right)\leq e^{-nh_p}\\
  &Pr\left(\oneover{n}\sum_{i=1}^nx_i-\bar{x}\leq -\delta\right)\leq e^{-nh_m}.
\end{align}
Eventually we sum the previous inequalities to obtain the law of large numbers 
with exponentially decreasing tails
\begin{equation}
  Pr\left(\left|\oneover{n}\sum_{i=1}^nx_i-\bar{x}\right|\geq \delta\right)\leq 
  e^{-nh_p}+e^{-nh_m}=O(e^{-n})=\epsilon.
\end{equation}
Observe that the small quantity $\epsilon>0$ obtained in this way, exponentially decreasing with increasing $n$, 
also depends on the difference parameter $\delta$ that we 
chose, as of course is to be expected. Indeed this dependence is implicit in the 
definition of $h_p$,$h_m$: by choosing $\delta$, we set different values of $a$ (for both the $s>0$ and $s<0$ cases)
and this in turn varies the point $s^*$ \eqref{zeropoint}, i.e. the range of values of $s$ ($s<s^*$ or $s>s^*$) which can be 
chosen to maximize the coefficient functions. In particular, since the 
expression \eqref{approxcoeff} is a small-s expansion, we do not know what the 
absolute supremum of $h(s)$ is and where it is located\footnote{The function $\mu_x(s)$ is convex, but we do not know its behaviour at large $s$.}. Thus by varying the 
range of $s$ accessible through the tuning of $\delta$, we may happen to exclude 
this and other local supremum points, resulting in (possibly discontinuously) varying values of $h_p$,$h_m$. 
\\
In any case, for our purpose we need only the existence of a range of values 
$s$, both above and below zero and depending on $\delta$, where the coefficient function $h(s)$ is 
positive, and this is guaranteed by the properties of the $\mu_x(s)$ function, 
respectively when $a>\bar{x}$ for positive $s$ and when $a<\bar{x}$ for negative 
$s$.

   \section{Proofs of Lemmas}\label{appendue}
    We give here the proofs of the remaining Lemmas of Section \ref{lemmas}.
    \begin{proof}[Proof of Lemma \ref{trdist}]
      For a hermitian operator we can always write $\omega=A-B$, where $A,B$ are positive matrices with disjoint supports, 
      representing $\omega$ respectively in the positive and negative part of 
      its support. Consider then the operator $\bar{\Lambda}=\Pi_A-\Pi_B$, with $\Pi_A$ 
      and $\Pi_B$ being projectors respectively on the support of $A$ and of $B$. 
      For this operator we can clearly state that 
      $-\mathbf{1}\leq\bar{\Lambda}\leq\mathbf{1}$, i.e.
      for all vectors $\ket{v}$ we have
      \begin{align}
        &\bra{v}(\Lambda-\mathbf{1})\ket{v}\leq0\\
        &\bra{v}(\Lambda+\mathbf{1})\ket{v}\geq 0.
      \end{align}
      By construction we obtain thus an operator which saturates the bound 
      \eqref{trdisteq}:
      \begin{align}
        \tr{\bar{\Lambda}\omega}&=\tr{\left(\Pi_A-\Pi_B\right)A}-\tr{\left(\Pi_A-\Pi_B\right)B}\\
        &=\tr{A}+\tr{B}=Tr|\omega|=\norm{\omega}_1.
      \end{align}
      In order to complete the proof, we need to show that $\bar{\Lambda}$ 
      is the maximizing operator among all possible
      $-\mathbf{1}\leq\Lambda\leq\mathbf{1}$. First observe, by diagonalising $A$ 
      and $B$, that
      \begin{align}
        &\tr{\Lambda 
        A}=\sum_k\alpha_k\bra{a_k}\Lambda\ket{a_k}\leq\sum_k\alpha_k\braket{a_k}{a_k}=\tr{A}\\
        &\tr{\Lambda 
        B}=\sum_k\beta_k\bra{b_k}\Lambda\ket{b_k}\geq\sum_k\beta_k(-\braket{b_k}{b_k})=-\tr{B}.
      \end{align}
      Thus 
      \begin{align}
        \tr{\Lambda \omega}=\tr{\Lambda A}-\tr{\Lambda B}\leq 
        \tr{A}+\tr{B}=Tr|\omega|=\norm{\omega}_1.
      \end{align}
    \end{proof}
    
    \begin{proof}[Proof of Lemma \ref{appclose}]
      Consider that
      \begin{align}
        2D(\rho,\sigma)&=\norm{\rho-\sigma}_1=\max_{-\mathbf{1}\leq\Lambda\leq\mathbf{1}}\tr{\Lambda(\sigma-\rho)}\\
        &\geq\tr{E(\sigma-\rho)},
      \end{align}
      which follows from applying Lemma \ref{trdist} and from the fact that $0\leq E\leq\mathbf{1}$ 
      surely is one of the operators included in the maximization procedure. The 
      result \eqref{appcloseq} is then easily obtained by separating the trace and rearranging terms in the 
      previous inequality.
    \end{proof}
    
    \begin{proof}[Proof of Lemma \ref{gentop}]
      Consider that
      \begin{align}
2D\left(\sqrt{E}\rho\sqrt{E},\rho\right)=\norm{\rho-\sqrt{E}\rho\sqrt{E}}_1&\leq\norm{\rho-\sqrt{E}\rho}_1+
   \norm{\sqrt{E}\rho-\sqrt{E}\rho\sqrt{E}}_1
  \\&=\norm{\left(\mathbf{1}-\sqrt{E}\right)\sqrt{\rho}\cdot\sqrt{\rho}}_1+
  \norm{\sqrt{E}\cdot\rho\left(\mathbf{1}-\sqrt{E}\right)}_1.\label{gentspezzato}
   \end{align}
     thanks to the triangular inequality for the trace distance. Now for the first term 
     write $\sqrt{\rho}$ in diagonal form $\{\sqrt{\lambda_k},\ket{f_k}\}$ and 
     use again the triangular inequality for the trace norm:
     \begin{align}
      \norm{\left(\mathbf{1}-\sqrt{E}\right)\sqrt{\rho}\sum_k\sqrt{\lambda_k}\ket{f_k}\bra{f_k}}_1 
      &\leq\sum_k\sqrt{\lambda_k}\norm{\left(\mathbf{1}-\sqrt{E}\right)\sqrt{\rho}\ket{f_k}\bra{f_k}}_1\\
     &=\sum_k\sqrt{\lambda_k}Tr\sqrt{\ket{f_k}\bra{f_k}\sqrt{\rho}\left(\mathbf{1}-\sqrt{E}\right)^2\sqrt{\rho}\ket{f_k}\bra{f_k}}\\
     &=\sum_k\sqrt{\lambda_k}\sqrt{\bra{f_k}\sqrt{\rho}\left(\mathbf{1}-\sqrt{E}\right)^2\sqrt{\rho}\ket{f_k}}.
     \end{align}
    Apply then the Cauchy-Schwarz inequality
    \begin{equation}
      |\vec{x}\cdot\vec{y}|^2\leq|\vec{x}|^2\cdot|\vec{y}|^2,
    \end{equation}
    with $x_k=\sqrt{\lambda_k}$ and 
    $y_k=\sqrt{\bra{f_k}\sqrt{\rho}\left(\mathbf{1}-\sqrt{E}\right)^2\sqrt{\rho}\ket{f_k}}$,
to obtain 
\begin{align}
  \norm{\left(\mathbf{1}-\sqrt{E}\right)\sqrt{\rho}\sum_k\sqrt{\lambda_k}\ket{f_k}\bra{f_k}}_1 
  &\leq\sqrt{\sum_k\lambda_k\sum_j\bra{f_j}\sqrt{\rho}\left(\mathbf{1}-\sqrt{E}\right)^2\sqrt{\rho}\ket{f_j}}\\
  &\leq\sqrt{\tr{\rho\left(\mathbf{1}-\sqrt{E}\right)^2}},
\end{align}
where we used the fact that $\tr{\rho}=\sum_k\lambda_k\leq1$.
For the second term in \eqref{gentspezzato} write instead $\sqrt{E}$ in its 
diagonal form $\{\sqrt{\nu_k},\ket{e_k}\}$ and proceed in a similar way as before:
\begin{align}
  \norm{\sum_k\sqrt{\nu_k}\ket{e_k}\bra{e_k}\rho\left(\mathbf{1}-\sqrt{E}\right)}_1
&\leq\sum_k\sqrt{\nu_k}\norm{\ket{e_k}\bra{e_k}\rho\left(\mathbf{1}-\sqrt{E}\right)}_1\\
&=\sum_k\sqrt{\nu_k}\norm{\left(\mathbf{1}-\sqrt{E}\right)\rho\ket{e_k}\bra{e_k}}_1\\
&=\sum_k\sqrt{\nu_k}\sqrt{\bra{e_k}\rho\left(\mathbf{1}-\sqrt{E}\right)^2\rho\ket{e_k}}\\
&\leq\sqrt{\sum_k\nu_k\sum_j\bra{e_j}\rho\left(\mathbf{1}-\sqrt{E}\right)^2\rho\ket{e_j}}\\
&\leq\sqrt{\tr{\rho^2\left(\mathbf{1}-\sqrt{E}\right)^2}}\\
&\leq\sqrt{\tr{\rho\left(\mathbf{1}-\sqrt{E}\right)^2}},
\end{align}
where we used the triangular inequality, the invariance of the trace norm under hermitian conjugation, the Cauchy-Schwarz inequality, the 
fact that $\tr{E}=\sum_k\nu_k\leq1$ and the property $\rho^2\leq\rho\leq\mathbf{1}$. 
 The inequality \eqref{gentspezzato} then simply becomes
 \begin{align}
   2D\left(\sqrt{E}\rho\sqrt{E},\rho\right)&\leq2\sqrt{\tr{\rho\left(\mathbf{1}-\sqrt{E}\right)^2}}\\
   &\leq 2\sqrt{\tr{\rho\left(\mathbf{1}-E\right)}},\label{quasifinale}
 \end{align}
 since 
 \begin{align}
   &0\leq E\leq\sqrt{E}\leq\mathbf{1}\\
   &\rightarrow 
   \left(\mathbf{1}-\sqrt{E}\right)^2=\mathbf{1}+E-2\sqrt{E}\leq\mathbf{1}-E.
 \end{align}
 Eventually we take the code average of \eqref{quasifinale} and use the 
 concavity of the square-root function and the hypotesis \eqref{gentopequno} to 
 obtain the thesis \eqref{gentopeqdue}:
 \begin{align}
   \ave{2D\left(\sqrt{E}\rho\sqrt{E},\rho\right)}&\leq2\sqrt{\ave{\tr{\rho\left(\mathbf{1}-E\right)}}}
   \leq2\sqrt{\epsilon}.
 \end{align}
  \end{proof}

\section{Derivation of the bisection POVM}\label{conca}
  Here we provide an explicit derivation of the POVM \eqref{bisegenerica} 
  associated with our bisection protocol. We consider each step to be carried 
  out as a unitary process on an enlarged system, consisting of the state $\ket{\Psi}$ 
  received by Bob (we take it pure for simplicity) and various ancillae, one for 
  each step. The ancillae start in a reference state $\ket{a}$ and will turn into one of three possible states depending on the result of the 
  measurement. In particular at the $u-$th step the ancilla state $\ket{0}$ ($\ket{1}$) 
  corresponds to having found the codeword in set  $\mathbf{C}^{(u)}_{{k_1},\ldots,k_{u-1},0}$ 
  ($\mathbf{C}^{(u)}_{{k_1},\ldots,k_{u-1},1}$), while the state $\ket{null}$ 
  corresponds to failure. \\
  We start by applying the first-step POVM 
  $\mathcal{M}^{(1)}=\{N^{(1)}_{0},N^{(1)}_{1},N^{(1)}_{null}\}$:
  \begin{align}
    U^{(1)}\left(\ket{\Psi}\ket{a}_1\right) &=\sqrt{N^{(1)}_0}\ket{\Psi}\ket{0}_1+\sqrt{N^{(1)}_1}\ket{\Psi}\ket{1}_1
    +\sqrt{N^{(1)}_{null}}\ket{\Psi}\ket{null}_1.
  \end{align}
  After the second step POVM $\mathcal{M}^{(2)}$ we obtain the state
  \begin{align}
    U^{(2)}\left(U^{(1)}\left(\ket{\Psi}\ket{a}_1\right)\ket{a}_2\right) &=
    \sqrt{N^{(2)}_{00}}\sqrt{N^{(1)}_0}\ket{\Psi}\ket{0}_1\ket{0}_2
    +\sqrt{N^{(2)}_{01}}\sqrt{N^{(1)}_0}\ket{\Psi}\ket{0}_1\ket{1}_2
    \nonumber\\&+\sqrt{N^{(2)}_{10}}\sqrt{N^{(1)}_1}\ket{\Psi}\ket{1}_1\ket{0}_2
    +\sqrt{N^{(2)}_{11}}\sqrt{N^{(1)}_1}\ket{\Psi}\ket{1}_1\ket{1}_2
    \nonumber\\&+\sqrt{N^{(2)}_{null}}\sqrt{N^{(1)}_0}\ket{\Psi}\ket{0}_1\ket{null}_2
    +\sqrt{N^{(2)}_{null}}\sqrt{N^{(1)}_1}\ket{\Psi}\ket{1}_1\ket{null}_2
    \nonumber\\&+\sqrt{N^{(1)}_{null}}\ket{\Psi}\ket{null}_1\ket{a}_2.
  \end{align}
  If we stop at this step, the probability of having found $\ket{\Psi}$ in
  a given set , e.g. $\mathbf{C}^{(2)}_{10}$, is
  \begin{align}
    P(10|\rho_{\Psi})&=\bra{\Psi}\sqrt{N^{(1)}_{1}}N^{(2)}_{10}\sqrt{N^{(1)}_{1}}\ket{\Psi}
    =\tr{\left|\sqrt{N^{(2)}_{10}} \sqrt{N^{(1)}_{1}} \right|^2 
    \rho_{\Psi}},
  \end{align}
  which corresponds to equation \eqref{bisegenerica} for $\vec{k}=(1,0)$ and can 
  be easily generalized to an arbitrary number of steps.


\begin{thebibliography}{99}
\bibitem{WildeBOOK} M. M. Wilde, {\it Quantum Information Theory} (Cambridge University Press 
2013).
   \bibitem{HolevoBOOK} A. S. Holevo, {\it Quantum Systems, Channels, Information} (de Gruyter Studies in Mathematical Physics, 2012).
  \bibitem{holevo1} A. S. Holevo, Probl. Peredachi Inf. \textbf{9}, 3 (1973); Probl. Inf. Transm. (Engl. Transl.) \textbf{9}, 110 (1973).
  \bibitem{holevo2} A. S. Holevo, IEEE Trans. Inf. Theory \textbf{44}, 269 (1998).
\bibitem{schumawest} B. Schumacher and M. D. Westmoreland, Phys. Rev. A \textbf{56}, 131 (1997); P. Hausladen, R. Jozsa, B. W. Schumacher,
M. Westmoreland, and W. K. Wootters, ibid. \textbf{54}, 1869 (1996).
\bibitem{hauswoot} P. Hausladen and W. K. Wooters, J. Mod. Opt. \textbf{41}, 2385 (1994).
\bibitem{holevo3} A. S. Holevo, e-print arXiv:quant-ph/9809023 [see also Tamagawa University Research Review, no. 4] (1998).
\bibitem{winter} A. Winter, IEEE Trans. Inf. Theory \textbf{45}, 2481 (1999).
\bibitem{oga}T. Ogawa, Ph.D. dissertation, University of Electro- Communications, Tokyo, Japan, 2000; (in Japanese) T. Ogawa and H. Nagaoka, in 
\textit{Proceedings of the 2002 IEEE International Symposium on Information Theory}, Lausanne, Switzerland, (IEEE, New, York, 2002), p. 73; T. Ogawa, IEEE Trans. Inf.
Theory \textbf{45}, 2486 (1999).
\bibitem{oganaga} T. Ogawa and H. Nagaoka, IEEE Trans. Inf. Theory \textbf{53}, 2261
(2007).
\bibitem{hayanaga} M. Hayashi and H. Nagaoka, IEEE Trans. Inf. Theory \textbf{49}, 1753
(2003).
\bibitem{hayashi} M. Hayashi, Phys. Rev. \textbf{A} 76, 062301 (2007); Commun. Math.
Phys. \textbf{289}, 1087 (2009).
\bibitem{seq1} S. Lloyd, V. Giovannetti, L. Maccone, Phys. Rev. Lett. \textbf{106}, 250501 
(2011).
\bibitem{seq2} V. Giovannetti, S. Lloyd and L. Maccone,  Phys. Lett. A 
\textbf{85}, 012302 (2012).
\bibitem{sen} P. Sen, e-print  	arXiv:1109.0802v1 [quant-ph] (2011).
 \bibitem{hastings} M. B. Hastings, Nat. Phys. \textbf{5}, 255 (2008).
 \bibitem{classicalinfo} T. M. Cover and J. A. Thomas, \textit{Elements of Information Theory} (Wiley, New York, 1991).
\bibitem{SchumaTyp} B. Schumacher, Phys. Rev. A \textbf{51}, 2738 (1995).
\bibitem{qhyptest} F. Hiai and D. Petz, Commun. Math. Phys. 143, 99 (1991);
T. Ogawa and H. Nagaoka, IEEE Trans. Inf. Theory 46, 2428
(2000).
\bibitem{wildeguha1} M. M. Wilde, SS. Guha, \textit{Proceedings of the 2012 International Symposium on Information Theory and its Applications}, 303-307 
(2012).
\bibitem{wildeguha2} M. M. Wilde, S. Guha, S.-H. Tan, S. Lloyd, \textit{Proceedings of the 2012 IEEE International Symposium on Information Theory} (ISIT 2012, Cambridge, MA, USA),  
551-555.
 \bibitem{HOLWER} A. S. Holevo and R. Werner, Phys. Rev. A {\bf 63}, 032312 (2001). 
 \bibitem{LOSSY} V. Giovannetti, S. Guha, S. Lloyd, L. Maccone, J. H. Shapiro, and H. P. Yuen, Phys. Rev. Lett. {\bf 92}, 027902 (2004). 
\bibitem{polarWildeGuha} M. M. Wilde, Saikat Guha, IEEE Trans. Inf. Theory \textbf{59}, 1175 (2013). 
\bibitem{Arikan} E. Arikan, IEEE Trans. Inf. Theory \textbf{55}, 3051 (2009).
\bibitem{wildeHayden} M. M. Wilde, O. Landon-Cardinal and P. Hayden, in {\it 8th Conference on the Theory of Quantum Computation, Communication and Cryptography (TQC 
2013)}, Dagstuhl, Germany, 2013 arXiv:1302.0398v1.
\bibitem{wildeRen1} M. M. Wilde and J. M. Renes, in {\it Proceedings of the 2012 International Symposium on Information Theory and its Applications},
 Honolulu, Hawaii, USA, October 2012. arXiv:1203.5794.
\bibitem{wildeRen2} M. M. Wilde and J. M. Renes, in {\it Proceedings of the 2012 International Symposium on Information Theory}, 
Boston, Massachusetts, USA, July 2012. arXiv:1201.2906.
\bibitem{privaQuant} J. M. Renes and M. M. Wilde, IEEE Trans. Inf. Theory  \textbf{60}, 3090 (2014).
\bibitem{wildeSen} M. M. Wilde, Proceedings of the Royal Society A \textbf{469}, 2157 
(2013).
\bibitem{NOTA} The code average in the first two properties can be removed when using 
a stronger notion of conditional typicality, which we do not state here for 
simplicity, since such average will appear quite naturally during calculations when making use of Shannon's averaging 
trick.
 \bibitem{cavesDrum} C. M. Caves and P. D. Drummond, Rev. Mod. Phys. \textbf{66}, 481 (1994).
\bibitem{braunstein} S. L. Braunstein and P. van Loock, Rev. Mod. Phys. \textbf{77}, 513 
(2005).
\bibitem{weedbrook} C. Weedbrook, S. Pirandola, R. Garcia-Patron, N. J. Cerf, T. C. Ralph, J. H. Shapiro and S. 
Lloyd, Rev. Mod. Phys. \textit{84}, 621 (2012).
      \bibitem{WINTERPHD} A. Winter,  PhD Thesis, Universit\"{a} Bielefled (1999). 
      \bibitem{HOLREV} A. S. Holevo and V. Giovannetti, Rep. Prog. Phys. {\bf 75}, 046001 (2012). 
    \bibitem{Gallager} R. G. Gallager, {\it Information Theory and Reliable Communication} (John Wiley \& Sons 1968). 
    \bibitem{NOTAgenSen} An alternative computation of the error probability can be carried out by employing a generalization of 
    Sen's non-commutative union bound\cite{wildeSen}, obtaining the same condition on the 
    single-step POVM for the optimality of the bisection protocol.
     \end{thebibliography}
 \end{document}